\documentclass[runningheads,a4paper]{llncs}

\usepackage[OT1]{fontenc}
\usepackage{fancyhdr}

\usepackage[pdftex, bookmarks=true, bookmarksopen=true]{hyperref}
\usepackage{amsmath}
\usepackage{amssymb}
\usepackage{tikz}
\usepackage{comment}
\usepackage{xspace}
\usepackage{varwidth}

\excludecomment{tacas}
\includecomment{techreport}
\newif\ifnewinterpolation
\newinterpolationtrue

\usetikzlibrary{decorations.pathreplacing, calc}%patterns}

\newcommand{\ceilfrac}[2]{\left\lceil \frac{#1}{#2} \right\rceil}
\newcommand{\floorfrac}[2]{\left\lfloor \frac{#1}{#2} \right\rfloor}
\newcommand{\proj}{\downharpoonright}

\newcommand{\euf}{$\mathcal{EUF}$\xspace}
\newcommand{\laq}{$\mathcal{LA}\left(\mathbb{Q}\right)$\xspace}
\newcommand{\laz}{$\mathcal{LA}\left(\mathbb{Z}\right)$\xspace}
\newcommand{\symb}{\mathop{\mathit{symb}}}
\newcommand{\eps}{\varepsilon}

\ifnewinterpolation
\renewcommand{\S}{}
\newcommand{\W}{}
\newcommand{\xor}{\mathrel{\mbox{xor}}}
\else
\renewcommand{\S}{^S}
\newcommand{\W}{^W}
\fi

\usepackage{mathpartir}

\hypersetup{colorlinks
  ,linkcolor=black
  ,filecolor=black
  ,urlcolor=black
  ,citecolor=black
  ,pdftitle={Proof Tree Preserving Interpolation},
  pdfauthor={J{\"u}rgen Christ, Jochen Hoenicke, Alexander Nutz},
  pdfkeywords={Craig interpolation }{SMT }{mixed literals }{theory combination }, 
  pdfsubject={Craig Interpolation in SMT }
}

\title{Proof Tree Preserving Interpolation
\thanks{This work is supported by the
  German Research Council (DFG) as part of the Transregional Collaborative
  Research Center ``Automatic Verification and Analysis of Complex Systems''
  (SFB/TR14 AVACS)}}
\global\toctitle={Proof Tree Preserving Interpolation}

\author{J{\"u}rgen Christ \and Jochen Hoenicke \and Alexander Nutz}
\institute{Chair of Software Engineering, University of Freiburg}

\begin{document}

\maketitle

\begin{abstract}
%  Craig interpolants are widely used in model checking and state space
%  abstraction. Interpolants typically are extracted from proofs produced by
%  theorem provers. While this extraction procedure is easy and well understood
%  in the context of propositional logic, extracting interpolants from a proof
%  generated by an SMT solver is more complex. In contrast to SAT solvers, SMT
%  solvers create new literals, e.g., to combine multiple theories in a
%  Nelson-Oppen style or to split the solution space using cuts. These literals
%  might contain symbols local to different parts of the interpolation
%  problem. Such literals are called \emph{mixed}, or, sometimes,
%  \emph{uncolourable}. Resolution steps on mixed literals are the major
%  difficulty when extracting interpolants from proofs from SMT
%  solvers.

  Craig interpolation in SMT is difficult because, e.\,g., theory combination
  and integer cuts introduce mixed literals, i.\,e., literals
  containing local symbols from both input formulae.  In this paper,
  we present a scheme to compute Craig interpolants in the presence of
  mixed literals.  Contrary to existing approaches, this scheme
  neither limits the inferences done by the SMT solver, nor does it
  transform the proof tree before extracting interpolants.  Our scheme
  works for the combination of uninterpreted functions and linear
  arithmetic but is extendable to other theories.  The
  scheme is implemented in the interpolating SMT solver SMTInterpol.
\end{abstract}

%----------------------------------------
\section{Introduction}
%----------------------------------------

A Craig interpolant for a pair of formulae $A$ and $B$ whose conjunction is
unsatisfiable is a formula $I$ that follows from $A$ and whose conjunction
with $B$ is unsatisfiable.  Furthermore, $I$ only contains symbols common
to $A$ and $B$. Model checking and state space abstraction~\cite{henzinger04afp,mcmillan06lai} make
intensive use of interpolation to achieve a higher degree of automation. This
increase in automation stems from the ability to fully automatically generate
interpolants from proofs produced by modern theorem provers. 

For propositional logic, a SAT solver typically produces resolution-based
proofs that show the unsatisfiability of an error path.
Extracting Craig interpolants from such
proofs is a well understood and easy task that can be accomplished, e.\,g.,
using the algorithms of Pudl\'ak~\cite{DBLP:journals/jsyml/Pudlak97} or
McMillan~\cite{DBLP:conf/tacas/McMillan04}.  An essential property of
the proofs generated by SAT solvers is that every proof step only involves
literals that occur in the input.

This property does not hold for proofs produced by SMT solvers for formulae
in a combination of first order theories.  Such solvers produce new literals
for different reasons.  First, to combine two theory solvers, SMT solvers
exchange (dis-)equalities between the symbols common to these two theories in
a Nelson-Oppen-style theory combination.
Second, various techniques dynamically generate new literals to simplify proof
generation. Third, new literals are introduced in the context of a
branch-and-bound or branch-and-cut search for non-convex theories. The
theory of linear integer arithmetic for example is typically solved by
searching a model for the relaxation of the formula to linear rational arithmetic
and then using branch-and-cut with Gomory cuts or \emph{extended
  branches}~\cite{Dillig2011} to remove the current non-integer solution 
from the solution space of the relaxation.

The literals produced by either of these techniques only contain symbols that are already present in the input.  However, a
literal produced by one of these techniques may be \emph{mixed}\footnote{Mixed
  literals sometimes are called \emph{uncolourable}.} in the sense that it may
contain symbols occurring only in $A$ and symbols occurring only in $B$.
These
literals pose the major difficulty when extracting interpolants from proofs
produced by SMT solvers.

%% While this extraction procedure is easy and well understood
%% in the context of propositional logic, extracting interpolants from a proof
%% generated by an SMT solver is more complex. In contrast to SAT solvers, SMT
%% solvers create new literals, e.\,g., to combine multiple theories in a
%% Nelson-Oppen style or to split the solution space using cuts. These literals
%% might contain symbols local to different parts of the interpolation
%% problem. Such literals are called \emph{mixed}, or, sometimes,
%% \emph{uncolourable}. Resolution steps on mixed literals are the major
%% difficulty when extracting interpolants from proofs from SMT solvers.

In this paper, we present a scheme to compute Craig interpolants in the
presence of mixed literals.  Our interpolation scheme is based on
syntactical restrictions of \emph{partial interpolants} and specialised rules to
interpolate resolution steps on mixed literals. 
This enables us to compute interpolants in the context of a state-of-the-art SMT
solver without manipulating the proof tree or restricting the solver in any
way.
We base our presentation
on the quantifier-free fragment of the combined theory of uninterpreted
functions and linear arithmetic over the rationals or the integers.
The interpolation scheme is used in the interpolating SMT solver
SMTInterpol~\cite{toolpaper}.  
\begin{tacas}
Proofs for the theorems in this paper are
given in the technical report~\cite{atr}.
\end{tacas}

%% theory of uninterpreted functions combined with the theory of linear
%% arithmetic either over the reals or the integers. We present an interpolation
%% scheme based on syntactic restriction of \emph{partial interpolants} and
%% specialised rules to interpolate resolution steps on mixed literals. Contrary
%% to existing approaches, this scheme neither limits the inferences done by the
%% SMT solver, nor does it transform the proof tree before extracting
%% interpolants. The interpolation scheme is used in the interpolating SMT solver
%% SMTInterpol.

%----------------------------------------
\subsubsection*{Related Work.}
%----------------------------------------

\begin{techreport}
Craig~\cite{Craig2011} shows in his seminal work on interpolation that for
every inconsistent pair of first order formulae an interpolant can be
derived. In the proof of the corresponding theorem he shows how to construct
interpolants without proofs by introducing quantifiers in the interpolant.
\end{techreport}
For Boolean circuits, Pudl\'ak~\cite{DBLP:journals/jsyml/Pudlak97} shows how
to construct quantifier-free interpolants from resolution proofs of
unsatisfiability.
\begin{techreport}

\end{techreport}
A different proof-based interpolation system is given by
McMillan~\cite{DBLP:conf/tacas/McMillan04} in his seminal paper on
interpolation for SMT. The presented method combines the theory of equality
and uninterpreted functions with the theory of linear rational
arithmetic. Interpolants are computed from partial interpolants by annotating
every proof step. The partial interpolants have a specific form that carries
information needed to combine the theories. The proof system is incomplete for
linear integer arithmetic as it cannot deal with arbitrary cuts and mixed
literals introduced by these cuts.

Brillout et al.~\cite{DBLP:conf/vmcai/BrilloutKRW11} present an interpolating
sequent calculus that can compute interpolants for the combination of
uninterpreted functions and linear integer arithmetic. The interpolants
computed using their method might contain quantifiers since they do not use
divisibility predicates. Furthermore their method limits the generation of
Gomory cuts in the integer solver to prevent some mixed cuts. The method
presented in this paper combines the two theories without quantifiers and,
furthermore, does not restrict any component of the solver.

Yorsh and Musuvathi~\cite{Yorsh2005} show how to combine interpolants
generated by an SMT solver based on Nelson-Oppen combination. They define the
concept of \emph{equality-interpolating theories}. These are theories that
can provide a shared term $t$ for 
a mixed literal $a=b$ that is derivable from an interpolation problem.
A troublesome mixed interface equality $a=b$ is rewritten into the
conjunction $a=t\land t=b$. They show that both, the theory of uninterpreted
functions and the theory of linear rational arithmetic are
equality-interpolating. We do not explicitly split the proof.
Additionally, our method can handle the
theory of linear integer arithmetic without any restriction on the solver. 
\begin{techreport}
The method of Yorsh and Musuvathi, however, cannot deal with cuts used by most
modern SMT solvers to decide linear integer arithmetic.
\end{techreport}

Cimatti et al.~\cite{Cimatti2010} present a method to compute interpolants for
linear rational arithmetic and difference logic. The method presented in this paper builds upon
their interpolation technique for linear rational arithmetic. For theories combined via delayed
theory combination, they show how to compute interpolants by
transforming a proof into a so-called \emph{ie-local} proof.  In these proofs,
mixed equalities are close to the leaves of the proof tree and
splitting them is cheap since the proof trees that have to be
duplicated are small. 
\begin{techreport}
A variant of this restricted search strategy is used by
MathSAT~\cite{Griggio2012} and CSIsat~\cite{DBLP:conf/cav/BeyerZM08}.
\end{techreport}

Goel et al.~\cite{Goel2009} present a generalisation of
equality-interpolating theories. They define the class of
\emph{almost-colourable proofs} and an algorithm to generate interpolants from
such proofs. Furthermore they describe a restricted DPLL system to generate
almost-colourable proofs. This system does not restrict the search if convex
theories are used. Their procedure is incomplete for non-convex theories like
linear arithmetic over integers since it prohibits the generation of mixed 
branches and cuts.

Recently, techniques to transform proofs gained a lot of
attention. Bruttomesso et al.~\cite{Bruttomesso2010} present a framework to
lift resolution steps on mixed literals into the leaves of the resolution
tree. Once a subproof only resolves on mixed literals, they replace this
subproof with the conclusion removing the mixed inferences. 
The newly generated lemmas however are mixed between different theories and require
special interpolation procedures. Even though these procedures only have to
deal with conjunctions of literals in the combined theories it is not obvious
how to compute interpolants in this setting. 
\begin{techreport}
Similar to our algorithm, they do
not restrict or interact with the SMT solver but take the proof as produced by
the solver. 
\end{techreport}
In contrast to our approach, they manipulate the proof in a way
that is worst-case exponential and rely on an interpolant generator for the
conjunctive fragment of the combined theories.

McMillan~\cite{McMillan12iZ3} presents a technique
to compute interpolants from Z3 proofs. 
Whenever a sub-proof contains mixed
literals, he 
extracts lemmas from the proof tree and delegates them to a second
(possibly slower) interpolating solver.

\begin{techreport}
For the theory of linear integer arithmetic \laz\ a lot of
different techniques were proposed.  Lynch et
al.~\cite{DBLP:conf/atva/LynchT08} present a method that produces interpolants
as long as no mixed cuts were introduced.  In the presence of such cuts, their
interpolants might contain symbols that violate the symbol condition of Craig
interpolants.

For linear Diophantine equations and linear modular equations, Jain et
al.~\cite{DBLP:journals/fmsd/JainCG09} present a method to compute linear
modular equations as interpolants. Their method however is limited to
equations and, thus, not suitable for the whole theory \laz.

Griggio~\cite{Griggio2011} shows how to compute interpolants for \laz\ based
on the \laz-solver from MathSAT~\cite{Griggio2012}. This solver uses
branch-and-bound and the cuts from proofs~\cite{Dillig2011} technique. Similar
to the technique presented by Kroening et
al.~\cite{DBLP:conf/lpar/KroeningLR10} the algorithm prevents generating mixed
cuts and, hence, restricts the inferences done by the solver.
\end{techreport}

\section{Preliminaries}

\begin{techreport}
In this section, we give an overview of what is needed to understand the
procedure we will propose in the later sections. We will briefly introduce the
logic and the theories used in this paper.
Furthermore, we define key terms like Craig interpolants and symbol sets.
\end{techreport}

\paragraph*{Logic, Theories, and SMT.}
We assume standard first-order logic. We operate within
the quantifier-free fragments of the theory of equality with uninterpreted
functions \euf and the theories of linear arithmetic over rationals \laq
and integers \laz.
The quantifier-free fragment of \laz is not closed under
interpolation.  Therefore, we augment the signature with
division by constant functions $\floorfrac{\cdot}{k}$ for all integers $k\geq 1$.

We use the standard notations $\models_T, \bot, \top$ to denote
entailment in the theory $T$, contradiction, and tautology.  In the
following, we drop the subscript $T$ as it always corresponds to the combined
theory of \euf, \laq, and \laz.

The literals in \laz are of the form $s \leq c$, where $c$ is an
integer constant and $s$ a linear combination of variables.  For \laq
we use constants $c \in \mathbb{Q}_\eps$, 
$\mathbb{Q}_\eps := \mathbb{Q} \cup \{q - \eps | q\in
\mathbb{Q}\}$ where the meaning of $s \leq q-\eps$ is $s < q$.  For
better readability we use, e.\,g., $x \leq y$ resp.\ $x > y$ to denote
$x-y\leq 0$ resp.\ $y-x \leq -\eps$.  In the integer case we use $x > y$ 
to denote $y-x \leq -1$.

Our algorithm operates on a proof of unsatisfiability generated
by an SMT solver based on DPLL$(T)$~\cite{dpllt}.
Such a proof is a resolution tree with the
$\bot$-clause at its root. The leaves of the tree are either clauses
from the input formulae\footnote{W.\,l.\,o.\,g.\ we assume input formulae are
  in conjunctive normal form.} or theory lemmas that are produced by one of
the theory solvers. The negation of a theory lemma is called a
\emph{conflict}.

The theory solvers for \euf, \laq, and \laz are working independently
and exchange (dis-)equality literals through the DPLL engine in a
Nelson-Oppen style~\cite{DBLP:journals/toplas/NelsonO79}.  Internally,
the solver for linear arithmetic uses only inequalities in
theory conflicts.  In the proof tree, the (dis-)equalities are related
to inequalities by the (valid) clauses $x=y \lor x<y \lor x>y$, and
$x\neq y \lor x\leq y$.  We call these leaves of the proof tree
\emph{theory combination clauses}.

\paragraph*{Interpolants and Symbol Sets.}

For a formula $F$, we use $\symb(F)$ to denote the set of non-theory
symbols occurring in $F$.  An interpolation problem is given by two
formulae $A$ and $B$ such that $A \land B \models \bot$.  
An interpolant of $A$ and $B$ is a formula $I$ such that 
(i) $A \models I$, (ii) $B \land I \models \bot$, and 
(iii) $\symb(I) \subseteq \symb(A)\cap \symb(B)$.

We call a symbol $s \in \symb(A)\cup\symb(B)$ \emph{shared}
if $s \in \symb(A) \cap \symb(B)$, \emph{$A$-local} if $s \in
\symb(A) \setminus \symb(B)$, and \emph{$B$-local} if $s \in
\symb(B) \setminus \symb(A)$.
Similarly, we call a term \emph{$A$-local} (\emph{$B$-local}) if it
contains at least one $A$-local ($B$-local) and no
$B$-local ($A$-local) symbols.
We call a term \emph{($AB$-)shared} if it contains only shared
symbols and \emph{($AB$-)mixed} if it contains $A$-local as well as $B$-local
symbols. The same terminology applies to formulae.

\paragraph*{Substitution in Formulae and Monotonicity.}
By $F[G_1]\ldots[G_n]$ we denote a formula in negation normal form with
sub-formulae $G_1,\ldots,G_n$ that occur positively in the formula.  Substituting
these sub-formulae by formula $G_1',\ldots,G_n'$ is denoted by $F[G_1']\ldots[G_n']$.  By $F(t)$ we
denote a formula with a sub-term $t$ that can appear anywhere in $F$.  The substitution of $t$ with a
term $t'$ is denoted by $F(t')$.

\begin{tacas}
The following lemma is important for the correctness proofs of our
interpolation scheme.
\end{tacas}
\begin{techreport}
The following lemma is important for the correctness proofs in the remainder
of this technical report. It also represents a concept that is important for
the understanding of the proposed procedure.
\end{techreport}

\begin{lemma}[Monotonicity]\label{lemma:monotonicity}
 Given a formula $F[G_1]\ldots[G_n]$ in negation normal form with
 sub-formulae $G_1,\ldots,G_n$ occurring only positively in
 the formula and formulae $G_1',\ldots,G_n'$, it holds that
 \[\left( \bigwedge_{i\in\{1,\ldots,n\}} (G_i \rightarrow G_i') \right)
   \rightarrow (F[G_1]\ldots[G_n] \rightarrow F[G_1']\ldots[G_n'])\]
\end{lemma}

\begin{techreport}
\begin{proof}
  We prove the claim by induction over the number of $\land$ and $\lor$
  connectives in $F[\cdot]\ldots[\cdot]$.  If $F[G_1]\ldots[G_n]$ is 
  a literal different from $G_1,\ldots,G_n$ the implication holds trivially.
  Also for the other base case $F[G_1]\ldots[G_n]\equiv G_i$ for some
  $i\in\{1,\dots,n\}$ the property holds. For the induction step
  observe that if $F_1[G_1]\ldots[G_n]\rightarrow  F_1[G'_1]\ldots[G'_n]$ and
  $F_2[G_1]\ldots[G_n]\rightarrow  F_2[G'_1]\ldots[G'_n]$, then 
  \begin{align*}
  F_1[G_1]\ldots[G_n] \land F_2[G_1]\ldots[G_n] &\rightarrow
  F_1[G'_1]\ldots[G'_n] \land F_2[G'_1]\ldots[G'_n] \text{ and }\\
  F_1[G_1]\ldots[G_n] \lor F_2[G_1]\ldots[G_n] &\rightarrow 
  F_1[G'_1]\ldots[G'_n] \lor F_2[G'_1]\ldots[G'_n].
  \end{align*}

  \vspace{-20pt}\strut\qed
\end{proof}
\end{techreport}

\section{Proof Tree-Based Interpolation}

Interpolants can be computed from proofs of unsatisfiability as
Pudl\'ak and McMillan have already shown.  In this section we will
introduce their algorithms.  Then, we will discuss the changes
necessary to handle mixed literals introduced, e.\,g., by theory combination.

\subsection{Pudl\'ak's and McMillan's Interpolation Algorithms}

Pudl\'ak's and McMillan's algorithms assume that the pivot literals are not mixed.
We will remove this restriction later.  We define a common
framework that is more general and can be instantiated to obtain Pudl\'ak's or McMillan's
algorithm to compute interpolants.  For this, we use two projection
functions on literals $\cdot\proj A$ and $\cdot\proj B$ as defined
below.
They have the properties (i)
$\symb(\ell\proj A)\subseteq\symb(A)$, (ii) $\symb(\ell\proj
B)\subseteq\symb(B)$, and (iii) $\ell \iff (\ell\proj A \land
\ell\proj B)$.  Other projection functions are possible and this
allows for varying the strength of the resulting interpolant as shown
in \cite{D'Silva2010}.  We extend the projection function to
conjunctions of literals component-wise.

\medskip
\centerline{
  \renewcommand{\arraystretch}{1.1}
 \begin{tabular}{l|c|c|c|c|}
    & \multicolumn{2}{c|}{Pudl\'ak}
    & \multicolumn{2}{c|}{McMillan}\\
    &$\ell\proj A$&$\ell\proj B$
    &$\ell\proj A$&$\ell\proj B$\\
   \hline
    $\ell$ is $A$-local
    & $\ell$ & $\top$ & $\ell$ & $\top$ \\
    $\ell$ is $B$-local
    & $\top$ & $\ell$ & $\top$ & $\ell$ \\
    $\ell$ is shared
    & $\ell$ & $\ell$ & $\top$ & $\ell$
  \end{tabular}
}
\medskip

Given an interpolation problem $A$ and $B$, a \emph{partial
interpolant} of a clause $C$ is an interpolant of the formulae $A
\land (\lnot C \proj A)$ and $B \land (\lnot C \proj B)$\footnote{Note that $\lnot C$ is a conjunction of literals. Thus, 
$\lnot C\proj A$ is well defined.}.
Partial interpolants can be computed inductively over the structure of
the proof tree.  A partial interpolant of a theory lemma $C$ can be
computed by a theory-specific interpolation routine as an interpolant
of $\lnot C \proj A$ and $\lnot C \proj B$. Note that the conjunction
is equivalent to $\lnot C$ and therefore unsatisfiable.  For an input
clause $C$ from the formula $A$ (resp.\ $B$), a partial interpolant is
$\lnot(\lnot C\setminus A)$ (resp.\ $\lnot C \setminus B$) where
$\lnot C\setminus A$ is the conjunction of all literals of $\lnot C$ 
that are not in $\lnot C \proj A$ and analogously for $\lnot C\setminus B$.  For a
resolution step, a partial interpolant can be computed using (\ref{rule:res}),
which is given below. For this rule, it is easy to show that $I_3$ is a partial
interpolant of $C_1\lor C_2$ given that $I_1$ and $I_2$ are partial
interpolants of $C_1\lor \ell$ and $C_2\lor \lnot \ell$, respectively. Note
that the ``otherwise'' case never triggers in McMillan's algorithm.
\begin{equation}\tag{rule-res}\label{rule:res}
\inferrule{C_1\lor \ell : I_1 \quad C_2\lor \lnot \ell : I_2}
            {C_1\lor C_2 : I_3} \quad 
	    \text{where }I_3 = \begin{cases}
	      I_1 \lor I_2 & \text{if }\ell\proj B = \top\\
	      I_1 \land I_2 & \text{if }\ell\proj A = \top\\
	      \begin{array}{l}(I_1\lor\ell) \land{}\\
	      (I_2\lor \lnot \ell)\end{array} & \text{otherwise}
	    \end{cases}
\end{equation}
As the partial interpolant of the root of the proof tree (which is
labelled with the clause $\bot$) is an interpolant of the input
formulae $A$ and $B$, this algorithm can be used to compute interpolants.

\begin{theorem}
 The above-given partial interpolants are correct, i.e., if
  $I_1$ is a partial interpolant of $C_1 \lor \ell$ 
  and $I_2$ is a partial interpolant of $C_2 \lor \lnot \ell$ 
  then $I_3$ is a partial interpolant of the  clause $C_1 \vee C_2$.
\end{theorem}

\begin{techreport}
\begin{proof}
 The third property, i.e., $\symb(I_3) \subseteq \symb(A) \cap \symb(B)$, 
 clearly holds
 if we assume it holds for $I_1$ and $I_2$. Note that in the
 ``otherwise'' case, $\ell$ is shared.
 We prove the other two partial interpolant properties separately.
 \subsubsection*{Inductivity.} We have to show
 \[A \land \lnot C_1\proj A \land \lnot C_2\proj A \models I_3.\]
 For this we use the inductivity of $I_1$
 and $I_2$:
 \begin{align*}
  & A \land \lnot C_1 \proj A \land \lnot \ell\proj A \models I_1 \tag{ind1} \\
  & A \land \lnot C_2 \proj A \land \ell\proj A \models I_2 \tag{ind2}
 \end{align*}  

 Assume $A$, $\lnot C_1\proj A$, and $\lnot C_2\proj A$. Then, (ind1)
 simplifies to $\lnot\ell \proj
 A \rightarrow I_1$ and (ind2) simplifies to $\ell \proj A \rightarrow I_2$.  We show that
 $I_3$ holds under these assumptions.

\paragraph{Case $\ell\proj B = \top$.}
 
 Then by the definition of the projection function, $\ell\proj A = \ell$ 
 and $\lnot \ell\proj A = \lnot \ell$ hold. If $\ell$ holds, (ind2) gives us
 $I_2$, otherwise (ind1) gives us $I_1$, thus $I_3 = I_1\lor I_2$ holds in both cases.
 
\paragraph{Case $\ell\proj A = \top$.}
 
 Then (ind1) gives us $I_1$ because $\lnot \ell\proj A = \top$ 
 (the negation of $\ell$ is still not in $A$), and (ind2) gives us $I_2$. 
 So $I_3 = I_1 \land I_2$ holds.
 
\paragraph{Case ``otherwise''.}
 
 By the definition of the projection function
 $\ell\proj A = \ell\proj B = \ell$ and 
 $\lnot \ell\proj A = \lnot \ell\proj B = \lnot \ell$. If $\ell$ holds,
 the left conjunct $(I_1 \lor \ell)$ of $I_3$ holds and the right
 conjunct $(I_2\lor \lnot\ell)$ of $I_3$ is fulfilled because (ind2) gives us $I_2$.
 If $\lnot \ell$ holds, (ind1) gives us $I_1$ and both conjuncts of $I_3$ hold.
 
\subsubsection*{Contradiction.} 
 We have to show:
 
 \[B \land \lnot C_1\proj B \land \lnot C_2\proj B \land I_3 \models \bot\]
 
 We use the contradiction properties of $I_1$ and $I_2$:
 \begin{align*}
  & B \land \lnot C_1\proj B \land \lnot \ell\proj B \land I_1 \models \bot \tag{cont1}\\
  & B \land \lnot C_2\proj B \land \ell \proj B \land I_2 \models \bot \tag{cont2}
 \end{align*}  

 If we assume $B$, $\lnot C_1\proj B$, and $\lnot C_2\proj B$, (cont1)
 simplifies to $\lnot \ell\proj B \land I_1 \rightarrow \bot$ and (cont2)
 simplifies to $\ell\proj B \land I_2 \rightarrow \bot$. We show $I_3
 \rightarrow \bot$.

 \paragraph{Case $\ell\proj B = \top$.}

 Then (cont1) and $\lnot \ell\proj B = \top$ give us $I_1 \rightarrow \bot$,
 and (cont2) and $\ell \proj B = \top$ give us $I_2 \rightarrow \bot$. 
 Thus $I_3 \equiv I_1 \lor I_2$ is contradictory.

 \paragraph{Case $\ell\proj A = \top$.}

 Then $\ell \proj B = \ell$ and $\lnot \ell \proj B = \lnot \ell$. Then, if 
 $\ell$ holds, (cont2) gives us $I_2 \rightarrow \bot$. If $\lnot \ell$ holds, 
(cont1) gives us
 $I_1 \rightarrow \bot$ analogously. In both cases,
 $I_3 \equiv I_1 \land I_2$ is contradictory.

 \paragraph{Case ``otherwise''.}
 
 By the definition of the projection function 
 $\ell\proj A = \ell\proj B = \ell$ and 
 $\lnot \ell\proj A = \lnot \ell\proj B = \lnot \ell$ hold.
 Assuming $I_3 \equiv (I_1\lor\ell) \land (I_2\lor \lnot\ell)$ holds, we prove
 a contradiction.
 If $\ell$ holds, the second conjunct of $I_3$ implies $I_2$.
 Then, (cont2) gives us a contradiction. 
 If $\lnot \ell$ holds, the first conjunct of $I_3$ implies $I_1$
 and (cont1) gives us a contradiction.
%, 
\qed 
\end{proof}
\end{techreport}

\subsection{Purification of Mixed Literals}\label{sec:purification}

The proofs generated by state-of-the-art SMT solvers may contain mixed literals. We tackle them by extending the projection functions to
these literals.  The problem here is that there is no projection
function that satisfies the conditions stated in the previous section.
Therefore, we relax the conditions by allowing fresh auxiliary
variables to occur in the projections.

We consider two different kinds of mixed literals: First, 
(dis-\nobreak)equalities of the form $a=b$ or $a\neq b$ for an $A$-local
variable $a$ and a $B$-local variable $b$ are introduced, e.\,g., by theory
combination or Ackermannization. Second, inequalities of the form $a + b \leq
c$ are introduced, e.\,g., by extended branches~\cite{Dillig2011} or bound
propagation. Here, $a$ is a linear combination of $A$-local variables,
$b$ is a linear combination of $B$-local and shared variables, and $c$ is a
constant.
\begin{techreport}
Adding the shared variable to the $B$-part is an arbitrary choice.  One gets
interpolants of different strengths by assigning some shared variables to the
$A$-part. It is only important to keep the projection of each literal
consistent throughout the proof.
\end{techreport}

\ifnewinterpolation
We split mixed literals using auxiliary variables, which we denote by
$x$ or $p_x$ in the following.  The variable $p_x$ has the type Boolean,
while $x$ has the same type as the variables in the literal.  One or
two fresh variables are introduced for each mixed literal.  We count
these variables as shared between $A$ and $B$.  The purpose of the
auxiliary variable $x$ is to capture the shared value that needs to be
propagated between $A$ and $B$.  When splitting a literal $\ell$ into
$A$- and $B$-part, we require that $\ell \Leftrightarrow \exists
x,p_x. (\ell\proj A) \land (\ell\proj B)$.  We need the additional
Boolean variable $p_x$ to split the literal $a\neq b$ into two (nearly)
symmetric parts.  This is achieved by the definitions below.
\begin{align*}
(a=b)\proj A & := (a=x) &
(a=b)\proj B & := (x=b) \\
(a\neq b)\proj A & := (p_x \xor a=x) &
(a\neq b)\proj B & := (\lnot p_x \xor x=b) \\
(a + b \leq c)\proj A & := (a + x\leq 0) &
(a + b \leq c)\proj B & := (-x + b\leq c)
\end{align*}
\else
We split mixed literals using auxiliary variables, which we denote by
$x$, $x_a$, or $x_b$ in the following.  One or two fresh variables are
introduced for each mixed literal.  We count these variables as shared
between $A$ and $B$.  The purpose of the auxiliary variables is to
capture the shared value that needs to be propagated between $A$ and $B$.
When splitting a literal $\ell$ into $A$- and $B$-part, we require that
$\ell \Leftrightarrow \exists x,x_a,x_b. (\ell\proj A) \land (\ell\proj B)$.
We need two variables $x_a$ and $x_b$ to split the literal $a\neq b$ into two
symmetric parts.  For symmetry we split the literal $a=b$ in the same fashion
instead of introducing only a single auxiliary variable.
This is achieved by the definitions below.
\begin{align*}
(a=b)\proj A & := (a=x_a \land x_a = x_b) &
(a=b)\proj B & := (x_a = x_b  \land x_b=b) \\
(a\neq b)\proj A & := (a=x_a \land x_a \neq x_b) &
(a\neq b)\proj B & := (x_a \neq x_b  \land x_b=b) \\
(a + b \leq c)\proj A & := (a + x\leq 0) &
(a + b \leq c)\proj B & := (-x + b\leq c)
\end{align*}
\fi
Since the mixed variables are considered to be shared, we allow them
to occur in the partial interpolant of a clause $C$.  However, a
variable may only occur if $C$ contains the corresponding literal.
This is achieved by a special interpolation rule for resolution steps
where the pivot literal is mixed.
The rules for the different mixed
literals are the core of our proposed algorithm and will be introduced
in the following sections.  

\ifnewinterpolation\else
Instead of with a single partial interpolant, we label each clause with a
pattern from which we can derive two partial interpolants, a strong and a weak
one.  The strong interpolant of a clause $C$ implies the weak interpolant under
the assumption that $\lnot C\proj A$ or $\lnot C\proj B$ holds.
Having two
interpolants enables us to complete the inductive proof.  
We show that the strong interpolant follows from the $A$-part of the
resolvent if the strong interpolants of the premises follow from their
respective $A$-part.  On the other hand, the weak interpolant is in
contradiction to the $B$-part in the resolvent if this is the case for
the premises.  Since the weak interpolant follows from the strong
interpolant this shows that both are partial interpolants.
The models for the strong and the weak interpolants only differ in the
values of the auxiliary variable.  The interpolants are needed because 
the ``right'' value
for the auxiliary variable is not known when interpolating the leaves 
of the proof tree.
The strong and the weak interpolant are identical if the clause does not
contain mixed literals.  Therefore, we derive only one interpolant for the
bottom clause.
\fi

\begin{techreport}
  \ifnewinterpolation
  \begin{lemma}[Partial Interpolation]\label{lem:weakstrongip}
    Given a mixed literal $\ell$ with auxiliary variable(s) $\vec x$
    and clauses $C_1\lor\ell$ and $C_2\lor\lnot\ell$ with corresponding 
    partial interpolants $I_1$ and $I_2$.  Let $C_3=C_1\lor C_2$ 
    be the result of a resolution step
    on $C_1\lor\ell$ and $C_2\lor\lnot\ell$ with pivot $\ell$.
    If a partial interpolant $I_3$ satisfies the symbol
    condition, and
    \begin{align*}
      &(\forall \vec x.\: (\lnot \ell \proj A \rightarrow I_1)  \land 
      (\ell \proj A \rightarrow I_2)) \rightarrow I_3 \tag{ind}\\
      &I_3 \rightarrow (\exists \vec x.\:
      (\lnot \ell \proj B \land I_1)  \lor
      (\ell \proj B \land I_2)) \tag{cont}
    \end{align*}
    then $I_3$ is a partial interpolant of $C_3$.
  \end{lemma}
  \else
  To concretise our setting we label a clause $C$ of the proof tree with an
  interpolation pattern $I$ from which we derive two interpolants $I\S$ and
  $I\W$.  The condition that the strong interpolant implies the weak
  interpolant can be expressed as
  \[ \lnot C\proj A \lor \lnot C\proj B \models I\S \rightarrow I\W.
  \tag{s-w}\]
  We will ensure this property when we define how the weak and strong
  interpolants are derived from an interpolant pattern.
  
  \begin{lemma}[Strong-Weak-Interpolation]\label{lem:weakstrongip}
    Given a mixed literal $\ell$ with auxiliary variable(s) $\vec x$
    and clauses $C_1\lor\ell$ and $C_2\lor\lnot\ell$ with corresponding 
    partial interpolant patterns $I_1$ resp.\ $I_2$, i.\,e., $I_1\S$ and
    $I_1\W$ (resp.\ $I_2\S$ and $I_2\W$) are partial interpolants of
    $C_1\lor \ell$ (resp.\ $C_2\lor\lnot\ell$).  Let $C_3=C_1\lor C_2$ 
    be the result of a resolution step
    on $C_1\lor\ell$ and $C_2\lor\lnot\ell$ with pivot $\ell$.
    If a partial interpolant pattern $I_3$ satisfies (s-w), the symbol
    condition, and
    \begin{align*}
      &(\forall \vec x.\: (\lnot \ell \proj A \rightarrow I_1\S)  \land 
      (\ell \proj A \rightarrow I_2\S)) \rightarrow I\S_3 \tag{ind}\\
      &I_3\W \rightarrow (\exists \vec x.\:
      (\lnot \ell \proj B \land I_1\W)  \lor
      (\ell \proj B \land I_2\W)) \tag{cont}
    \end{align*}
    then $I_3\S$ and $I_3\W$ are partial interpolants of $C_3$.
  \end{lemma}
  \fi

  \begin{proof}
    We need to show inductivity and contradiction for the partial interpolants.
    \subsubsection*{Inductivity.}
    For this we use inductivity of $I_1\S$ and $I_2\S$:
    \begin{align*}
      & A \land \lnot C_1 \proj A \land \lnot \ell\proj A \models I_1\S \\
      & A \land \lnot C_2 \proj A \land \ell\proj A \models I_2\S
    \end{align*}
    Since $\vec x$ does not appear in $C_1\proj A$, $C_2 \proj A$ nor $A$,
    we can conclude
    \begin{align*}
      & A \land \lnot C_1 \proj A \models \forall \vec x.\: \lnot \ell\proj A \rightarrow I_1\S \\
      & A \land \lnot C_2 \proj A \models \forall \vec x.\: \ell\proj A \rightarrow I_2\S
    \end{align*}
    Combining these and pulling the quantifier over the conjunction gives
    \begin{align*}
      & A \land \lnot C_1 \proj A \land \lnot C_2 \proj A 
       \models \forall \vec x.\: (\lnot \ell\proj A \rightarrow I_1\S) \land
                                (\ell\proj A \rightarrow I_2\S)
    \end{align*}
    Using (ind), this shows that inductivity for  $I_3\S$ holds:
    \[A \land \lnot C_1\proj A \land \lnot C_2\proj A \models I_3\S.\]
    \ifnewinterpolation\else
    Using (s-w) we immediately get inductivity for $I_3\W$:
    \[A \land \lnot C_1\proj A \land \lnot C_2\proj A \models I_3\W.\]
    \fi

    \subsubsection*{Contradiction.}
    First, we show the contradiction property for $I_3\W$:
    \[B \land \lnot C_1 \proj B \land \lnot C_2 \proj B \land I_3\W \models \bot.\]
    Assume the formulae on the left-hand side hold.  From (cond) we can conclude that there is some $\vec x$ such that 
    \[(\lnot \ell \proj B \land I_1\W)  \lor
      (\ell \proj B \land I_2\W)\]
    If the first disjunct is true we can derive the contradiction
    using the contradiction property of $I_1\W$:
    \begin{align*}
      & B \land \lnot C_1 \proj B \land \lnot \ell\proj B \land I_1\W \models \bot 
    \end{align*}  
    Otherwise, the second disjunct holds and we can use the contradiction
    property of $I_2\W$
    \begin{align*}
      & B \land \lnot C_2 \proj B \land \ell\proj B  \land I_2\W \models \bot
    \end{align*}  
    This shows the contradiction property for $I_3\W$.
    \ifnewinterpolation\else
    Using (s-w) we immediately get the contradiction property for $I_3\S$: 
    \[B \land \lnot C_1 \proj B \land \lnot C_2 \proj B \land I_3\S \models
    \bot.\]

    \vspace{-19pt}\strut
    \fi\qed
  \end{proof}
\end{techreport}

It is important to state here that the given purification of a
literal into two new literals is not a modification of the proof tree
or any of its nodes. The proof tree would no longer be well-formed if
we replaced a mixed literal by the disjunction or conjunction of the
purified parts.  The purification is only used to define partial
interpolants of clauses.  In fact, it is only used in the
correctness proof of our method and is not even done explicitly in the
implementation.

\begin{techreport}

\subsection{Lemma Used in the Correctness Proof}

The following lemma will help us prove the correctness of our proposed new 
interpolation rules.

%\paragraph{Notation} In the following we denote by $F[\phi]$ a formula %TODO in die preliminaries
%in negation normal form (NNF) with a placeholder $\phi$ that occurs
%positive in the formula.  Substituting this placeholder by a formula
%$G$ is denoted by $F[G]$.  By $F(x)$ we denote a formula with a
%placeholder $x$ that stands for an arbitrary subterm of some literal
%in $F(x)$.  The substitution of $x$ with a term $t$ is denoted by $F(t)$.

\begin{lemma}[Deep Substitution]
  Let $F_1[G_{11}]\dots[G_{1n}]$ and $F_2[G_{21}]\dots[G_{2m}]$ be two formulae with sub-formulae $G_{1i}$ for $1 \leq i \leq n$ and $G_{2j}$ for $1 \leq j \leq m$ occurring positively in $F_1$ and $F_2$.

    If $\bigwedge_{i\in \{1,\dots,n\}} \bigwedge_{j\in \{1,\dots,m\}} G_{1i} \land G_{2j} \rightarrow G_{3ij}$ 
    holds, then
    \begin{align*}
      &F_1[G_{11}]\dots [G_{1n}] \land F_2[G_{21}]\dots [G_{2m}] 
      \rightarrow\\& F_1[F_2[G_{311}]\dots[G_{31m}]]\dots
                   [F_2[G_{3n1}]\dots[G_{3nm}]].
    \end{align*}

\end{lemma}
\begin{proof}
  \begin{align*}
   &\hphantom{{}\Rightarrow}
    \bigwedge_{i\in \{1,\dots,n\}} \bigwedge_{j\in \{1,\dots,m\}} 
     ((G_{1i} \land G_{2j}) \rightarrow G_{3ij})\\
    &\Leftrightarrow
    \bigwedge_{i\in \{1,\dots,n\}} \bigwedge_{j\in \{1,\dots,m\}} 
     (G_{1i} \rightarrow (G_{2j} \rightarrow G_{3ij}))\\
    &\Leftrightarrow
    \bigwedge_{i\in \{1,\dots,n\}} (G_{1i} \rightarrow 
      \bigwedge_{j\in \{1,\dots,m\}} (G_{2j} \rightarrow G_{3ij}))\\
    \text{\{monotonicity\}} &\Rightarrow
    \bigwedge_{i\in \{1,\dots,n\}} (G_{1i} \rightarrow 
      (F_2[G_{21}]\dots [G_{2m}] \rightarrow F_2[G_{3i1}]\dots [G_{3im}]))\\
    &\Leftrightarrow
    \bigwedge_{i\in \{1,\dots,n\}} (F_2[G_{21}]\dots [G_{2m}] \rightarrow
     (G_{1i} \rightarrow 
      F_2[G_{3i1}]\dots [G_{3im}]))\\
    &\Leftrightarrow
    (F_2[G_{21}]\dots [G_{2m}] \rightarrow
    \bigwedge_{i\in \{1,\dots,n\}} (G_{1i} \rightarrow 
      F_2[G_{3i1}]\dots [G_{3im}]))\\
    \text{\{monotonicity\}}&\Rightarrow
    (F_2[G_{21}]\dots [G_{2m}] \rightarrow (F_1[G_{11}]\dots [G_{1n}]
    \rightarrow \\
     & \qquad
    F_1[F_2[G_{311}]\dots [G_{31m}]]\dots [F_2[G_{3n1}]\dots
      [G_{3nm}]]))\\
    &\Leftrightarrow
     (F_1[G_{11}]\dots [G_{1n}] \land F_2[G_{21}]\dots [G_{2m}]) \rightarrow \\
       &\qquad F_1[F_2[G_{311}]\dots [G_{31m}]]\dots [F_2[G_{3n1}]\dots [G_{3nm}]]))
  \end{align*}
  \vspace{-12pt}\strut\qed
\end{proof}

\end{techreport}

%----------------------------------------
\section{Uninterpreted Functions}
%----------------------------------------
In this section we will present the part of our algorithm that is
specific to the theory \euf.  The only mixed atom that is considered
by this theory is $a=b$ where $a$ is $A$-local and $b$ is $B$-local.

\subsection{Leaf Interpolation}

The \euf solver is based on the congruence closure
algorithm~\cite{Detlefs2005}. The theory lemmas are generated from
conflicts involving a single disequality that is in contradiction to a
path of equalities.  Thus, the clause generated from such a conflict
consists of a single equality literal and several disequality
literals.

When computing the partial interpolants of the theory lemmas, we
internally split the mixed literals according to
Section~\ref{sec:purification}.  Then we use an algorithm similar
to~\cite{Fuchs2009} to compute an interpolant.  This algorithm
basically summarises the $A$-equalities that are adjacent on the path
of equalities.  

\ifnewinterpolation
If the theory lemma contains a mixed equality $a=b$ (without negation),
%with auxiliary variable $x_a,x_b$, 
it corresponds to the single
disequality in the conflict.  
This disequality is split into $p_x \xor a=x$ and $\lnot p_x \xor x= b$
and the resulting interpolant depends on the value of $p_x$.  If $p_x=\bot$,
the disequality is part of the $B$-part and $x$ is the
end of an equality path summing up the equalities from $A$.
Thus, the computed interpolant contains a literal of the form $x=s$.
%% This leads to a literal $x_a=s$ in
%% the interpolant.  
If $p_x=\top$, then the $A$-part of the literal is $a \neq x$, and the
resulting interpolant contains the literal $x\neq s$ instead.
Thus, the resulting interpolant can be put into the form $I[p_x \xor x=s]$.
Note that the formula $p_x\xor x=s$ occurs positively in the
interpolant and is the only part of the interpolant containing $x$ and $p_x$.
We define
\begin{align*}
  EQ(x,s) &:= (p_x \xor x = s)
\end{align*}
and require that the partial interpolant of a clause containing the
literal $a=b$ always has the form $I[EQ(x,s)]$ where $x$ and $p_x$ do not
occur anywhere else.

For theory lemmas containing the literal $a\neq b$, the corresponding
auxiliary variable $x$ may appear anywhere in the partial
interpolant, even under a function symbol.  A simple example is the
theory conflict $s\neq f(a) \land a=(x =)b \land f(b)=s$,
which has the partial interpolant $s\neq f(x)$.  In general the 
partial interpolant of such a clause has the form $I(x)$.
\fi

When two partial interpolants for clauses containing $a=b$ are
combined using~(\ref{rule:res}), i.\,e., the pivot literal is a
non-mixed literal but the mixed literal $a=b$ occurs in $C_1$ and
$C_2$, the resulting partial interpolant may contain $EQ(x,s_1)$ and
$EQ(x,s_2)$ for different shared terms $s_1, s_2$.  In general, we allow
the partial interpolants to have the form $I[EQ(x,s_1)]\dots[EQ(x,s_n)]$.

\subsection{Pivoting of Mixed Equalities}

\begin{tacas}
We require that every clause containing $a=b$ with auxiliary variables
$x_a,x_b$ is always labelled with a formula of the form
$I[EQ(x,s_1)]\dots[EQ(x,s_n)]$ and that this is a partial
interpolant of the clause for both $EQ\S$ and $EQ\W$.  As discussed
above, this is automatically the case for the theory lemmas computed
from conflicts in the congruence closure algorithm.  This property is
also preserved by (\ref{rule:res}) and this rule also preserves the
property of being a 
\ifnewinterpolation\else strong or weak \fi
partial interpolant.
\end{tacas}
\begin{techreport}
We require that every clause $C$ containing $a=b$ with auxiliary variables
\ifnewinterpolation $x,p_x$ \else $x_a,x_b$ \fi
is always labelled with a formula of the form
$I[EQ(x,s_1)]\dots[EQ(x,s_n)]$.  
\ifnewinterpolation\else
From this pattern we get the strong
resp.\ weak interpolant by substituting $EQ\S(x,s_i)$ resp.\ $EQ\W(x,s_i)$
($i\in\{1,\dots,n\}$) in $I$.  To show (s-w), assume $\lnot C\proj
A\lor \lnot C\proj B$.
Since $\lnot C$ contains $a \neq b$ we can derive $x_a\neq x_b$.
Then, $EQ\S(x,s_i) \rightarrow EQ\W(x,s_i)$ holds and by
monotonicity we get
\[
I[EQ\S(x,s_1)]\dots[EQ\S(x,s_n)] \rightarrow
I[EQ\W(x,s_1)]\dots[EQ\W(x,s_n)].\]
\fi
As discussed above, the partial interpolants computed for
conflicts in the congruence closure algorithm are of the form
$I[EQ(x,s_1)]\dots[EQ(x,s_n)]$. 
This property is also preserved by (\ref{rule:res}), and by Theorem~1 this
rule also preserves the property of being a 
\ifnewinterpolation\else strong or weak \fi
partial interpolant.
\end{techreport}
\ifnewinterpolation\else

\fi
\begin{tacas}
On the other hand, a clause containing the literal $a\neq b$ is
labelled with a formula of the form $I(x)$, i.\,e., the auxiliary
variable $x$ can occur at
arbitrary positions.  Both $I(x_a)$ and $I(x_b)$ are partial
interpolants of the clause.  Again, the form $I(x)$ and the property
of being a partial interpolant is also preserved by (\ref{rule:res}).
\end{tacas}
\begin{techreport}
On the other hand, a clause containing the literal $a\neq b$ is
labelled with a formula of the form $I(x)$, i.\,e., the auxiliary
variable $x$ can occur at 
arbitrary positions.  
\ifnewinterpolation\else
The strong resp.\ weak interpolants are derived
from this pattern as $I(x_a)$ resp.\ $I(x_b)$. 
To show (s-w), assume again $\lnot C\proj A\lor \lnot C\proj B$.
In this case $\lnot C$ contains $a = b$, so we can derive $x_a= x_b$.
Then, $I(x_a) \rightarrow I(x_b)$ holds.
\fi
Again, the form $I(x)$ and the property
of being a partial interpolant is also preserved by (\ref{rule:res}).
\end{techreport}

We use the following rule to interpolate the resolution step on the mixed
literal $a=b$.
%%  At a resolution step where the mixed
%% literal $a=b$ is pivoted, we use the following interpolation rule.
%
\begin{equation}\label{rule:inteq}
\inferrule {C_1\lor a=b:I_1[EQ(x,s_1)]\dots[EQ(x,s_n)] \\ C_2 \lor a\neq b: I_2(x) } 
 {C_1 \lor C_2: I_1[I_2(s_1)]\dots[I_2(s_n)] }
 \tag{rule-eq}
\end{equation}
The rule replaces every literal $EQ(x,s_i)$ in $I_1$ with the formula
$I_2(s_i)$, in which every $x$ is substituted by $s_i$. Therefore,
the auxiliary variable introduced for the mixed literal $a=b$ is removed.

\begin{theorem}[Soundness of (\ref{rule:inteq})]
  Let $a=b$ be a mixed literal with auxiliary variable $x$.  If
  $I_1[EQ(x,s_1)]\dots[EQ(x,s_n)]$ 
  \ifnewinterpolation is a partial interpolant 
  \else yields two (strong and weak) partial interpolants \fi
  of $C_1 \lor a=b$ and $I_2(x)$
  \ifnewinterpolation a partial interpolant 
  \else two partial interpolants \fi
  of $C_2 \lor a\neq b$ then
  $I_1[I_2(s_1)]\dots[I_2(s_n)]$ 
  \ifnewinterpolation is a partial interpolant 
  \else yields two partial interpolants \fi
  of the clause $C_1 \vee C_2$.
\end{theorem}
%\begin{tacas}
%The proof is given in the technical report~\cite{atr}.
%\end{tacas}

\begin{techreport}
\begin{proof}
The symbol condition for $I_1[I_2(s_1)]\dots[I_2(s_n)]$ clearly holds
if we assume that it holds for $I_1[EQ(x,s_1)]\dots[EQ(x,s_n)]$ and $I_2(x)$. 
\ifnewinterpolation\else
By construction of weak and strong interpolants (s-w) holds. 
\fi
Hence, after
we show (ind) and (cont), we can apply Lemma~\ref{lem:weakstrongip}.

\ifnewinterpolation
\subsubsection*{Inductivity.}
We assume
\begin{align*}
\forall x,p_x.~ &((p_x \xor a = x) \rightarrow I_1[p_x \xor x=s_1]\dots[p_x \xor x=s_n])\\
&\land (a = x \rightarrow I_2(x))
\end{align*}
and show $I_1[I_2(s_1)]\dots[I_2(s_n)]$.
Instantiating $x:= s_i$ for all $i\in\{1,\dots,n\}$
and taking the second conjunct gives
$\bigwedge_{i\in\{1,\dots,n\}} (a=s_i \rightarrow I_2(s_i))$.
Instantiating $p_x:=\bot$ and $x:=a$ and taking the first conjunct gives
$I_1[a=s_1]\dots[a=s_n]$.  
With monotonicity we get
$I_1[I_2(s_1)]\dots[I_2(s_n)]$ as desired.

\subsubsection*{Contradiction.}
We have to show
\begin{align*}
I_1[I_2(s_1)]&\ldots[I_2(s_n)] \rightarrow \\
&\exists x,p_x.~(((\lnot p_x \xor x=b) \land I_1[p_x \xor x = s_1]\dots[p_x \xor x = s_n])\\
&\hphantom{\exists x,p_x.~}\lor (x=b \land I_2(x)))
\end{align*}
We show the implication for $p_x:=\top$ and $x:=b$.  It simplifies to
\[ I_1[I_2(s_1)]\dots[I_2(s_n)] \rightarrow 
I_1[b \neq s_1]\dots[b \neq s_n] \lor I_2(b)
\]
If $I_2(b)$ holds the implication is true.
If $I_2(b)$ does not hold, we have 
\[\bigwedge_{i\in\{1,\dots,n\}} (I_2(s_i) \rightarrow b \neq s_i)\]
With monotonicity we get $I_1[I_2(s_1)]\dots[I_2(s_n)]
\rightarrow I_1[b\neq s_1]\dots[b\neq s_n]$.
\qed
\fi
\end{proof}
\end{techreport}

%%  LocalWords:  disequality interpolants interpolant EQ inductivity

\subsection{Example}

We demonstrate our algorithm on the following example:
\begin{align*}
  A\equiv&(\lnot q \lor a=s_1)\land(q \lor a=s_2)\land f(a)=t\\
  B\equiv&(\lnot q \lor b=s_1)\land(q \lor b=s_2)\land f(b)\neq t
\end{align*}
The conjunction $A\land B$ is unsatisfiable.  In this example, $a$ is
$A$-local, $b$ is $B$-local and the remaining symbols are shared.

Assume the theory solver for \euf introduces the mixed literal $a=b$ and
provides the lemmas (i) $f(a)\neq t\lor a\neq b\lor 
f(b)=t$, (ii) $a\neq s_1\lor b\neq s_1\lor a=b$, and (iii) $a\neq
s_2\lor b\neq s_2\lor a=b$.  
Let the variable $x$ be
associated with the equality $a=b$.  Then, we label the lemmas with (i)
$f(x)=t$, (ii) $EQ(x,s_1)$, and (iii) $EQ(x,s_2)$.

We compute an interpolant for $A$ and $B$ using Pudl\'ak's algorithm.  Since
the input is already in conjunctive normal form, we can directly apply resolution.
\begin{tacas}
  From lemma (ii) and the input clauses $\lnot q \lor a=s_1$ and $\lnot q \lor
  b=s_1$ we can derive the clause $\lnot q\lor a=b$.  The partial interpolant
  of the derived clause is still $EQ(x,s_1)$, since the partial interpolants
  of the input clauses are $\bot$ resp.\ $\top$.  Similarly, from 
  lemma (iii) and the input clauses $q \lor a=s_2$
  and $q \lor b=s_2$ we can derive the clause $q\lor a=b$ with partial
  interpolant $EQ(x,s_2)$.  A resolution step on these two clauses with $q$ as
  pivot yields the clause $a=b$. Since $q$ is a shared literal, Pudl\'ak's
  algorithm introduces 
  the case distinction.  Hence, we get the partial interpolant
  $(EQ(x,s_2)\lor q)\land(EQ(x,s_1)\lor \lnot q)$.  Note that this interpolant
  has the
  form $I_1[EQ(x,s_1)][EQ(x,s_2)]$ and, therefore, satisfies the syntactical
  restrictions required by our scheme.

  From the \euf-lemma (i) and the input clauses $f(a)=t$ and $f(b)\neq t$, we
  can derive the clause $a\neq b$ with partial interpolant $f(x)=t$.  Note that
  this interpolant has the form $I_2(x)$ which also corresponds to the syntactical
  restrictions needed for our method.

  If we apply the final resolution step on the mixed literal $a=b$ using
  (\ref{rule:inteq}), we get the
  interpolant $I_1[I_2(s_1)][I_2(s_2)]$ which corresponds to the interpolant
  $(f(s_2)=t\lor q)\land(f(s_1)=t\lor \lnot q)$.
\end{tacas}
\begin{techreport}
  Note that for Pudl\'ak's algorithm every input clause has the partial
  interpolant $\bot$ ($\top$) if it is part of $A$ ($B$).  In the following
  derivation trees we apply the following simplifications without explicitely
  stating them:
  \begin{align*}
    F\land\top&\equiv F\\
    F\lor\bot&\equiv F
  \end{align*}

  From lemma (ii) and the input clauses $\lnot q \lor a=s_1$ and $\lnot q \lor
  b=s_1$ we can derive the clause $\lnot q\lor a=b$.  The partial interpolant
  of the derived clause is still $EQ(x,s_1)$.
  \[
  \inferrule*{ \inferrule*{
    \lnot q \lor a=s_1 : \bot \\
    a\neq s_1\lor b\neq s_1\lor a=b : EQ(x,s_1)}
            {b\neq s_1\lor \lnot q \lor a=b : EQ(x,s_1)}\\
            b=s_1\lor\lnot q : \top}
            {\lnot q\lor a=b : EQ(x,s_1)}
  \]
  Similarly, from 
  lemma (iii) and the input clauses $q \lor a=s_2$
  and $q \lor b=s_2$ we can derive the clause $q\lor a=b$ with partial
  interpolant $EQ(x,s_2)$.
  \[
  \inferrule*{ \inferrule*{
    q \lor a=s_2 : \bot \\
    a\neq s_2\lor b\neq s_2\lor a=b : EQ(x,s_2)}
            {b\neq s_2\lor q \lor a=b : EQ(x,s_2)}\\
            b=s_2\lor q : \top}
            {q\lor a=b : EQ(x,s_2)}
  \]
  A resolution step on these two clauses with $q$ as
  pivot yields the clause $a=b$. Since $q$ is a shared literal, Pudl\'ak's
  algorithm introduces 
  the case distinction.  Hence, we get the partial interpolant
  $(EQ(x,s_2) \lor q)\land(EQ(x,s_1) \lor \lnot q)$.  Note that this interpolant
  has the
  form $I_1[EQ(x,s_1)][EQ(x,s_2)]$ and, therefore, satisfies the syntactical
  restrictions.
  \[
  \inferrule*
      { q\lor a=b : EQ(x,s_2) \\ \lnot q\lor a=b : EQ(x,s_1) }
      {a=b : (EQ(x,s_2) \lor q)\land(EQ(x,s_1) \lor \lnot q)}
  \]
  From the \euf-lemma (i) and the input clauses $f(a)=t$ and $f(b)\neq t$, we
  can derive the clause $a\neq b$ with partial interpolant $f(x)=t$.  Note that
  this interpolant has the form $I_2(x)$ which also corresponds to the syntactical
  restrictions needed for our method.
  \[
  \inferrule*
  {
  \inferrule*
  { f(a)=t : \bot \\ f(a)\neq t\lor a\neq b\lor f(b)=t : f(x)=t }
  { f(b)=t \lor a\neq b : f(x)=t } \\ f(b)\neq t : \top }
  { a\neq b : f(x)=t }
  \]
  If we apply the final resolution step on the mixed literal $a=b$ using
  (\ref{rule:inteq}), we get the
  interpolant $I_1[I_2(s_1)][I_2(s_2)]$ which corresponds to the interpolant
  $(f(s_2)=t\lor q)\land(f(s_1)=t \lor \lnot q)$.
  \[
  \inferrule*
  { a=b : (EQ(x,s_2) \lor q)\land(EQ(x,s_1) \lor \lnot q) \\ a\neq b : f(x)=t }
  { \bot : (f(s_2)=t \lor q)\land(f(s_1)=t \lor \lnot q) }
  \]
  When resolving on $q$ in the derivations above, the mixed literal $a=b$
  occurs in both antecedents.  This leads to the form
  $I[EQ(x,s_1)][EQ(x,s_2)]$.  We can prevent this by resolving in a different
  order.
  We could first resolve the clause $q\lor a=b$ with the clause $a\neq b$ and
  obtain the partial interpolant $f(s_2)=t$ using (\ref{rule:inteq}).
  \[
  \inferrule*{ a=b \lor q : EQ(x,s_2) \\ a\neq b : f(x) = t }
            { q : f(s_2) = t }
  \]
  Then we could resolve the clause $\lnot q\lor a=b$ with the clause $a\neq b$ and
  obtain the partial interpolant $f(s_1)=t$ again using (\ref{rule:inteq}).
  \[
  \inferrule*{ a=b \lor \lnot q : EQ(x,s_1) \\ a\neq b : f(x) = t }
    { \lnot q : f(s_1) = t }
  \]
  The final resolution step on $q$ will then introduce the case distinction
  according to Pudl\'ak's algorithm.  This results in the same interpolant.
  \[
  \inferrule*
  { q : f(s_2)=t \\
    \lnot q : f(s_1)=t }
  { \bot : (f(s_2)=t \lor q)\land(f(s_1)=t \lor \lnot q) }
  \]
\end{techreport}

%----------------------------------------
\section{Linear Real and Integer Arithmetic}
%----------------------------------------

Our solver for linear arithmetic is based on a variant of the Simplex
approach~\cite{Dutertre2006}.  A theory conflict is a conjunction of
literals $\ell_j$ of the form $\sum_i a_{ij} x_{i} \leq b_j$.  The
proof of unsatisfiability is given by Farkas coefficients $k_j\geq 0$
for each inequality $\ell_j$.  These coefficients have the properties
$\sum_j k_ja_{ij} = 0$ and $\sum_j k_jb_j < 0$.  In the following we
use the notation of adding inequalities (provided the coefficients are
positive).  Thus, we write $\sum_j k_j \ell_j$ for $\sum_i (\sum_j k_j
a_{ij}) x_i \leq \sum_j k_jb_j$. With the property of the Farkas
coefficients we get a contradiction ($0<0$) and this shows that the theory
conflict is unsatisfiable.

A conjunction of literals may have rational but no integer solutions.
In this case, there are no Farkas coefficients that can prove the
unsatisfiability.  So for the integer case, our solver may introduce
extended branches~\cite{Dillig2011}, which are just branches of the DPLL engine
on newly introduced literals.  In the proof tree this results in resolution
steps with these literals as pivots.

%% -*- latex-mode -*-
\begin{example}
  The formula $t \leq 2a \leq r \leq 2b+1 \leq t$ has no integer
  solution but a rational solution.  Introducing the branch $a\leq b
  \lor b < a$ leads to the theory conflicts $t\leq 2a \leq 2b \leq
  t-1$ and $r \leq 2b+1 \leq 2a-1 \leq r-1$ (note that $b<a$ is
  equivalent to $b+1 \leq a$).  The corresponding proof tree is given
  below.  The Farkas coefficients in the theory lemmas are given in
  parenthesis.  Note that the proof tree shows the clauses, i.\,e.,
  the negated conflicts.  A node with more than two parents denotes that
  multiple applications of the resolution rule are taken one after another.

  \centerline{\begin{tikzpicture}\footnotesize
    \node (t2) at(0,0) {$\begin{array}{r}\lnot(r\le 2b+1)\;(\cdot 1)\\\lnot
        (b+1\le a)\;(\cdot 2)\\\lnot(2a\le r)\;(\cdot 1)\end{array}$};
    \node (t1) at(6,0) {$\begin{array}{r}
                          \lnot(t\le 2a)\;(\cdot 1)\\
                          \lnot(a\le b)\;(\cdot 2)\\
			  \lnot(2b+1\le t)\;(\cdot 1)
			  \end{array}$};
    \node (i21) at(2.2,.4) {$r\leq 2b+1$};
    \node (i22) at(3.4,-.1) {$2a\leq r$};
    \node (i11) at(8.2,.4) {$t\leq 2a$};
    \node (i12) at(9.4,-.1) {$2b+1\leq t$};
    \node (r2) at(2.2,-1) {$a\le b$};
    \node (r1) at(8.2,-1) {$\lnot(a\le b)$};
    \node (b) at(5.2,-1.7) {$\bot$};
    \draw (t1)-- (r1);
    \draw (i11)-- (r1);
    \draw (i12)-- (r1);
    \draw (t2)-- (r2);
    \draw (i21)-- (r2);
    \draw (i22)-- (r2);
    \draw (r1)-- (b);
    \draw (r2)-- (b);
  \end{tikzpicture}}

Now consider the problem of deriving an interpolant between $A\equiv t\le
2a\le r$ and $B\equiv r\le 2b+1\le t$.  We can obtain an interpolant by
annotating the above resolution tree with partial interpolants.  
To compute a partial interpolant for the theory lemma
$\lnot(r\le 2b+1) \lor \lnot(b+1\le a) \lor \lnot(2a\le r)$,
we purify the \emph{negated} clause according to the definition in
Section~\ref{sec:purification}, which gives
\[ r\le 2b+1  \land  x_1 \le a \land -x_1 + b + 1 \le 0 \land 2a\le r. \]
Then, we sum up the $A$-part of the conflict (the second
and fourth literal) multiplied by their corresponding Farkas coefficients.
This yields the interpolant $2x_1 \le r$.
Similarly, the negation of the theory lemma 
$\lnot(t\le 2a) \lor \lnot(a\le b) \lor  \lnot(2b+1\le t)$ is purified to
\[t\le 2a \land x_2+a\le 0 \land -x_2\le b \land  2b+1\le t,\]
which yields the partial interpolant $2x_2+t \le 0$.  
Note, that we have to introduce different variables for each literal.
Intuitively, the variable $x_1$ stands for $a$ and $x_2$ for $-a$.  
Using Pudl\'ak's algorithm we can derive the same interpolants for 
the clause $a\leq b$ resp.\ $\lnot (a\leq b)$.

For the final resolution step, the two partial interpolants 
$2x_1 \le r$ and $2x_2+t \le 0$ are combined into
the final interpolant of the problem.
Summing up these inequalities with $x_1=-x_2$ we get $t\le r$.  
While this follows from $A$, it is not
inconsistent with $B$.  We need an additional argument that, given $r=t$,
$r$ has to be an even integer.  This also follows from the partial
interpolants when setting $x_1=-x_2$: $t\leq -2x_2 = 2x_1 \leq r$.
\ifnewinterpolation
The final interpolant computed by our algorithm is 
$t \leq 2\floorfrac{r}{2}$.
\else
The final interpolant computed by our algorithm is 
$t \leq r \land (t \geq r \rightarrow t \leq 2\lfloor r/2\rfloor)$.
\fi

%Using the Strong-Weak-Interpolation lemma we have to find I3 such that
%(ALL x_1 x_2.   x_1 \le a --> 2x_1 \le r  
%           &&  x_2 +a \le 0 --> 2x_2 + t \le 0) --> I3
%and
%I3 --> (EX x_1 x_2. -x_1 + b + 1 \le 0 && 2x_1 \le r  
%                  ||  -x_2  \le b && 2x_2 + t \le 0)

In general, we can derive additional constraints on the variables
if the constraint resulting from summing up the two partial interpolants
holds very tightly. We know implicitly that $x_1=-x_2$ is an integer
value between $t/2$ and $r/2$.  If $t$ equals
$r$ or almost equals $r$ there are only a few possible values which we can
explicitly express using the division function as in the example above.
\ifnewinterpolation
We assume that the (partial) interpolant $F$ always has a certain
property.  There is some term $s$ and some constant $k$, such that
for $s > 0$ the interpolant is always false and for $s < -k$ the
interpolant is always true (in our case $s = t-r$ and $k = 0$). 
For a partial interpolant that still contains auxiliary variables $\vec x$,
we additionally require that $s$ contains them with a positive coefficient
and that $F$ is monotone on $\vec x$, i.\,e., $\vec x \geq \vec x'$ implies
$F(\vec x) \rightarrow F(\vec x')$.
\else
This leads to the general form $t-r\leq 0\land (t-r\geq
-k\rightarrow F)$.  In our example we have $k=0$ and $F$ specifies
that $r=t$ is even.
\fi
\end{example}

%%  LocalWords:  interpolant interpolants

%TODO as in \cite{Dillig2011}.
%FIXME end

%\begin{tikzpicture}[thick]
%    \draw(0,0)--(12.5,0);
%    \foreach \x/\xtext in {0/0,3/1,6/2,9/3,12/4}
%      \draw(\x,0pt)--(\x,3pt) node[above] {\xtext};
%    \foreach \x in {0,1,...,12}
%      \draw(\x,0pt)--(\x,-3pt) node[below] {$\frac{\x}{3}$};
%\end{tikzpicture}

% (minimal) template-picture for illustrating the eq-stuff
%\begin{tikzpicture}
% \draw(0,0)--(10,0);
% \node [above] (s2) at (2,0) {$\frac{s_2}{c_2}$};
% \node [above] (s1) at (8,0) {$-\frac{s_1}{c_1}$};
% \draw let \p1 = (s2) in (\x1,-3pt)--(\x1,3pt);
% \draw (8,-3pt)--(8,3pt);
% \draw [decoration={brace,mirror},
%    decorate] 
%    (6,0) -- (8,0) node[pos=0.5,below] {$\frac{k_1}{c_1}$}; 
%\end{tikzpicture}

%in integer context, $\varepsilon$ means 1, usually... %TODO genau festlegen

To mechanise the reasoning used in the example above, our resolution
rule for mixed inequality literals requires that the interpolant patterns
that label the clauses have a certain shape.  An auxiliary
variable of a mixed inequality literal may only occur in the 
interpolant pattern
if the negated literal appears in the clause.  Let $\vec x$ denote the
set of auxiliary variables that occur in the pattern.  We require 
that these
variables only occur inside a special sub-formula of the form $LA(s(\vec x), k,
F(\vec x))$.  The first parameter $s$ is a linear term over the
variables in $\vec x$ and arbitrary other terms not involving $\vec
x$.  The coefficients of the variables $\vec x$ in $s$ must all be
positive.  The second parameter $k\in \mathbb{Q}_\eps$ is a
constant value.  In the real case we only allow the values $0$ and
$-\eps$. In the integer case we allow $k\in\mathbb{Z}, k\geq -1$.  
To simplify the presentation, we sometimes write $-\eps$ for $-1$ in the integer case.
The third parameter $F(\vec x)$ is a formula that contains the variables
from $\vec x$ at arbitrary positions.
\ifnewinterpolation
We require that $F$ is monotone, i.\,e., $\vec x \geq \vec x'$ implies 
$F(\vec x) \rightarrow F(\vec x')$.  
Moreover, $F(\vec x)=\bot$ for $s(\vec x) > 0$ and
$F(\vec x)=\top$ for $s(\vec x) < -k$.  The sub-formula
$LA(s(\vec x), k, F(\vec x))$ stands for $F(\vec x)$ and it is only
used to remember what the values of $s$ and $k$ are.

The intuition behind the formula $LA(s(\vec x), k, F(\vec x))$ is that
$s(\vec x) \leq 0$ summarises the inequality chain that follows from the $A$-part of
the formula.  On this chain there may be some constraints on
intermediate values.  In the example above the $A$-part contains the
chain $t \leq 2a \leq r$, which is summarised to $s \leq 0$ (with $s=t-r$).
Furthermore the $A$-part implies that there is an even integer value between
$t$ and $r$.  If $s < -k$ (with $k=0$ in this case),  $t$ and $r$ are 
distinct, and there always is an even integer between them.
However, if $-k \leq s \leq 0$, the truth value of the interpolant depends on 
whether $t$ is even. 
\fi

In the remainder of the section, we will give the interpolants for the
leaves produced by the linear arithmetic solver and for the resolvent of the
resolution step where the pivot is a mixed linear inequality.

\subsection{Leaf Interpolation}

As mentioned above, our solver produces for a clause $C\equiv\lnot
\ell_1\lor \dots \lor\lnot \ell_m$ some Farkas coefficients
$k_1,\dots,k_m\geq 0$ such that $\sum_j k_j \ell_j$ yields a
contradiction $0 < 0$.  A partial interpolant for a theory lemma can be
computed by summing up the $A$-part of the conflict: $I$ is defined as
$\sum_j k_j (\ell_j \proj A)$~(if $\ell_j\proj A = \top$ we regard it
as $0\leq 0$, i.\,e., it is not added to the sum).  It is a valid interpolant as it
clearly follows from $\lnot C\proj A \iff \ell_1\proj A \land \dots
\land \ell_m\proj A$.  Moreover, we have that $I + \sum_j
k_j(\ell_j\proj B)$ yields $0< 0$, since for every literal, even
for mixed literals, $\ell_j\proj A + \ell_j\proj B = \ell_j$
holds\footnote{Strictly speaking this does not hold for shared literals, where
  $\ell\proj A = \ell\proj B = \ell$.  In that case use $k_j=0$ in
$I+\sum_j k_j(\ell_j\proj B)$ to see that $I$ is indeed a partial interpolant.}.  This
shows that $I\land \lnot C\proj B$ is unsatisfiable.

The linear constraint $\sum_j k_j (\ell_j \proj A)$ can be expressed
as $s(\vec x) \leq 0$.  Thus, we can equivalently write this
interpolant in our pattern as $LA(s(\vec x), -\eps,
\ifnewinterpolation s(\vec x) \leq 0 \else \bot\fi)$.  Since the
%JC: Changed -1 to -\eps in above line
Farkas coefficients are all positive and the auxiliary variables
introduced to define $\ell \proj A$ for mixed literals contain $x$
positively, the resulting term $s(\vec x)$ will also always contain
$x$ with a positive coefficient.

\paragraph{Theory combination lemmas.}

As mentioned in the preliminaries, we use theory combination clauses
to propagate equalities from and to the Simplex core of the linear
arithmetic solver.  These clauses must also be labelled with partial interpolants.
\begin{techreport}
In the following we give interpolants for those theory combination lemmas.
We will start with the case where no mixed literals occur, and treat lemmas
containing mixed literals afterwards.
\end{techreport}
\begin{tacas}The interesting case is when these clauses contain mixed literals.%, see
Table~\ref{tbl:iplthcomb} shows the corresponding partial interpolants.
The non-mixed case is given in the technical report.
    
The interpolant for the clause $a=b\lor a< b \lor a> b$ deserves more
explanation.  This clause is used to propagate equalities from the
linear arithmetic solver if it can derive $a\leq b$ and $b\leq a$.  In the
interpolant, $x_1$ is the variable with $b \leq x_1 \leq a$, and $x_2$
the variable with $a \leq -x_2 \leq b$.  The formula
$LA(x_1+x_2,0,EQ(x,x_1))$ basically states that $x_1 \leq -x_2$ and
that if $x_1\geq -x_2$ then $x_1$ equals the shared value $x$ of the
equality $a=b$.  We stress that the interpolant has the required form:
$x_1$ and $x_2$ only occur inside an $LA$ and with the correct
coefficients in $x_1+x_2$ while $x$ only occurs as first parameter of an
$EQ$ term, which appears positively in the negation normal form (by
the definition of $LA\S$ and $LA\W$).

\begin{table}[t]
  \begin{minipage}{.5\textwidth}
  Clause $C$: $a \neq b \lor a\leq b$\\
  $\lnot C\proj A$: $\ifnewinterpolation a=x \else a=x_a\land x_a=x_b\fi
                     \land -a+x_1\leq 0$\\
  $\lnot C\proj B$: $\ifnewinterpolation x=b \else x_a=x_b\land x_b=b\fi
                     \land -x_1+b< 0$\\
  Interpolant $I$: $LA(-x+x_1, -\eps, 
                    \ifnewinterpolation x_1 \leq x\else \bot\fi)$
  \end{minipage}%
  \begin{minipage}{.5\textwidth}
  Clause $C$: $a \neq b \lor a\geq b$\\
  $\lnot C\proj A$: $\ifnewinterpolation a=x \else a=x_a\land x_a=x_b\fi
                     \land a+x_2 \leq 0$\\
  $\lnot C\proj B$: $\ifnewinterpolation x=b \else x_a=x_b\land x_b=b\fi
                     \land -x_2-b< 0$\\
  Interpolant $I$: $LA(x+x_2, -\eps, 
                    \ifnewinterpolation x \leq -x_2\else \bot\fi)$
  \end{minipage}\\[6pt]
  \centerline{\begin{minipage}{.8\textwidth}
  Clause $C$: $a = b \lor a < b \lor a > b$\\
  $\lnot C\proj A$: $\ifnewinterpolation p_x \xor a=x 
                     \else a=x_a\land x_a \neq x_b\fi
                     \land -a+x_1\leq 0 \land a+x_2\leq 0$\\
  $\lnot C\proj B$: $\ifnewinterpolation \lnot p_x \xor x=b 
                     \else x_a \neq x_b\land x_b = b\fi
                    \land -x_1+b\leq 0 \land -x_2 - b\leq 0$\\
  Interpolant $I$: $LA(x_1+ x_2, 0, EQ(x, x_1))$
  \end{minipage}}\\%Just a little bit more space

  \caption{Interpolation of mixed theory combination clauses. We
    assume $a$ is $A$-local, $b$ is $B$-local, $a-b\leq 0$ has the
    auxiliary variable $x_1$, $b-a\leq 0$ has the auxiliary variable
    $x_2$ and $a=b$ the auxiliary variables $x_a$ and $x_b$. \label{tbl:iplthcomb}}
\end{table}
\end{tacas}

\begin{techreport}
\paragraph*{Interpolation of Non-Mixed Theory Combination Lemmas.}

If a theory combination lemma $t=u \lor t< u \lor t> u$ or $t \neq u \lor
t\leq u$ contains no mixed literal, we can compute partial interpolants as
follows.  If all literals in the clause are $A$-local, the formula $\bot$ is a
partial interpolant.  If all literals are $B$-local, the formula $\top$ is a
partial interpolant.  These are the same interpolants Pudl\'ak's algorithm
would give for input clauses from $A$ resp.\ $B$.

Otherwise, one of the literals belongs to $A$ and one to $B$.  The
symbols $t$ and $u$ have to be shared between $A$ and $B$ since they
appear in all literals.  We can derive a partial interpolant by
conjoining the negated literals projected to the $A$ partition.
\begin{align*}
  I &\equiv (t\neq u)\proj A \land (t\geq u)\proj A \land (t \leq u) \proj A. 
  &&\quad \mbox{for }t=u \lor t< u \lor t> u\\
  I &\equiv (t = u)\proj A \land (t > u) \proj A
  &&\quad \mbox{for }t\neq u \lor t\leq u
\end{align*}

Since we defined $I$ as $\lnot C \proj A$, the first property of the partial
interpolant holds trivially.  Also $I \land \lnot C \proj B$ is equivalent to
$\lnot C$ and therefore false.  The symbol condition is satisfied as $t$ and
$u$ are shared symbols.
\bigskip

\paragraph*{Interpolation of AB-Mixed Theory Combination Lemmas.}

\begin{table}[t]
  \begin{varwidth}{.5\textwidth}
  Clause $C$: $a \neq b \lor a\leq b$\\
  $\lnot C\proj A$: $\ifnewinterpolation a=x \else a=x_a\land x_a=x_b\fi
                     \land -a+x_1\leq 0$\\
  $\lnot C\proj B$: $\ifnewinterpolation x=b \else x_a=x_b\land x_b=b\fi
                     \land -x_1+b< 0$\\
  Interpolant $I$: $LA(-x+x_1, -\eps, 
                    \ifnewinterpolation x_1 \leq x\else \bot\fi)$
  \end{varwidth}\hfill
  \begin{varwidth}{.5\textwidth}
  Clause $C$: $a \neq b \lor a\geq b$\\
  $\lnot C\proj A$: $\ifnewinterpolation a=x \else a=x_a\land x_a=x_b\fi
                     \land a+x_2 \leq 0$\\
  $\lnot C\proj B$: $\ifnewinterpolation x=b \else x_a=x_b\land x_b=b\fi
                     \land -x_2-b< 0$\\
  Interpolant $I$: $LA(x+x_2, -\eps, 
                    \ifnewinterpolation x \leq -x_2\else \bot\fi)$
  \end{varwidth}\\[6pt]
  \centerline{\begin{varwidth}{.8\textwidth}
  Clause $C$: $a = b \lor a < b \lor a > b$\\
  $\lnot C\proj A$: $\ifnewinterpolation (p_x \xor a=x)
                     \else a=x_a\land x_a \neq x_b\fi
                     \land -a+x_1\leq 0 \land a+x_2\leq 0$\\
  $\lnot C\proj B$: $\ifnewinterpolation (\lnot p_x \xor x=b)
                     \else x_a \neq x_b\land x_b = b\fi
                    \land -x_1+b\leq 0 \land -x_2 - b\leq 0$\\
  Interpolant $I$: $LA(x_1+ x_2, 0, 
  \ifnewinterpolation x_1\leq -x_2 \land (x_1\geq -x_2 \rightarrow EQ(x, x_1))
  \else EQ(x, x_1)\fi)$
  \end{varwidth}}\\%Just a little bit more space
  \caption{Interpolation of mixed theory combination clauses. We
    assume $a$ is $A$-local, $b$ is $B$-local, $a-b\leq 0$ has the
    auxiliary variable $x_1$, $b-a\leq 0$ has the auxiliary variable
    $x_2$ and $a=b$ the auxiliary variables 
    \ifnewinterpolation $x$ and $p_x$\else $x_a$ and $x_b$\fi.
    \label{tbl:iplthcomb}}
\end{table}

If we are in the mixed case, all three literals are mixed. One of the two
terms must be $A$-local (in the following we denote this term by $a$) the
other term $B$-local (which we denote by $b$). To purify the literals, we introduce a fresh auxiliary
variable for each literal. Table~\ref{tbl:iplthcomb} depicts all possible mixed
theory lemmas together with the projections $\lnot C \proj A$ and $\lnot C
\proj B$ and a partial interpolant of the clause.

\begin{lemma}
The interpolants shown in Table~\ref{tbl:iplthcomb} are correct
partial interpolants of their respective clauses.
\end{lemma}

\begin{proof}
  First, we convince ourselves that these interpolants are of the right form:  The
  variables $x_1$ and $x_2$ appear in the first parameter of $LA$ with positive
  coefficients.  For the first two clauses that contain the literal $a\neq b$,
  the interpolant is allowed to contain $x$ at arbitrary positions.  
  \ifnewinterpolation Note that in the first interpolant
  $x_1\leq x$ is false for $-x+x_1 > 0$ and true for $-x+x_1 <\eps$, 
  i.\,e., $-x+x_1 \leq 0$.  Also, $x_1 \geq x_1'$ implies 
  $x_1\leq x \rightarrow x_1' \leq x$.  Similarly, for the second interpolant.

  In the
  third clause,
  $F(x_1,x_2) = x_1\leq -x_2 \land (x_1\geq -x_2 \rightarrow EQ(x, x_1))$ 
  is false for $x_1 + x_2 > 0$ (because of the first conjunct) and true 
  for $x_1 + x_2 < 0$ (because the implication holds vacuously).
  Also, $x_1 \geq x_1'$ and $x_2 \geq x_2'$ implies 
  $F(x_1,x_2) \rightarrow F(x_1', x_2')$.  To see this, note that
  $F(x_1,x_2)$ is false if $x_1' \geq -x_2'$ and $x_1' \neq x_1$.
  The variable $x$ appears only in an $EQ$-term which occurs
  positively in the partial interpolant. 
  \else  
  In the
  third clause the variable $x$ appears only in an $EQ$-term which occurs
  positively in the expanded form of the partial interpolant. 
  \fi

  Next we show
  \begin{align*}
    &\lnot C \proj A \models I\S \tag{Inductivity}\\
    &\lnot C \proj B \land I\W \models \bot \tag{Contradiction}
  \end{align*}

  \subsubsection*{Inductivity.}  
  \ifnewinterpolation
  For the clause $a\neq b\lor a\leq b$, the interpolant follows 
  from $\lnot C\proj A$, as $a=x$ and $-a + x_1 \leq 0$
  imply $x_1\leq x$. Similarly for the clause $a\neq b\lor a\geq b$,
  $\lnot C\proj A$ contains
  $a=x$ and $a + x_2 \leq 0$, which implies $x \leq -x_2$.
  \else
  For the clauses $a\neq b\lor a\leq b$ and $a\neq b\lor a\geq b$ the argument
  is symmetric.  We show only the case $C \equiv a\neq b \lor a\leq b$. 
  Assume $\lnot C\proj A$ and show $LA\S(-x_a+x_1,-\eps,\bot)$:
  \[
  \forall x_1' \leq x_1.~ -x_a+x_1' \leq 0 \land 
  (-x_a+x_1' \geq \eps \rightarrow \bot)
  \]
  Let $x_1'\leq x_1$. From $\lnot C \proj A$ we have $-a+x_1\leq 0$ and
  $a=x_a$. Hence, $-x_a + x_1' \leq 0$ as desired, which also shows that the
  implication in the second conjunct holds vacuously.

  \fi
  Now consider the clause $a= b\lor a< b \lor b< a$.
  \ifnewinterpolation
  Here, $\lnot C\proj A$ implies $x_1\leq -x_2$ and 
  that if $x_1\geq -x_2$ holds, then $x_1 = a =-x_2$. 
  Hence, $x_1\leq -x_2 \land x_1 \geq -x_2 \rightarrow EQ(x, x_1)$ holds.
  \else
  We assume $\lnot C\proj A$ and show $LA\S(x_1+x_2,0,x_a=x_1)$:
  \[\forall x_1' \leq x_1, x_2' \leq x_2.~ 
  x_1'+x_2' \leq 0 \land (x_1'+x_2' \geq 0 \rightarrow x_a = x_1').\]
  Let $x_1'\leq x_1$ and $x_2'\leq x_2$. 
  From $\lnot C\proj A$ we have $x_1 \leq a$ and $a+x_2\leq 0$.
  Thus, $x_1'+x_2' \leq x_1+x_2 \leq a + (-a) = 0$.
  Moreover, if $x_1'+x_2' \geq 0$ then
  \[ a\leq -x_2 \leq -x_2' \leq x_1' \leq x_1 \leq a,\]
  hence $a=x_1$.  With $a=x_a$ (also a part of $\lnot C \proj A$),
  this yields $x_a=x_1$.
  This shows that $LA\S(x_1+x_2,0,x_a=x_1)$ holds.
  \fi
  
  \subsubsection*{Contradiction.}
  Again we only show the first and third case. For the clause
  $C \equiv a\neq b \lor a\leq b$,
  \ifnewinterpolation
  note that $\lnot C\proj B$ and $LA\W(-x+x_1,-\eps,x_1 \leq x)$ 
  give the contradiction $x_1 > b = x > x_1$.
  For the clause $C \equiv a= b\lor a< b \lor b< a$,
  $\lnot C\proj B$ implies $x_1 \geq b \geq -x_2$. 
  With $x_1\leq -x_2$ from the interpolant this gives $x_1 = b$. Also,
  $x_1\geq -x_2\rightarrow EQ(x, x_1)$ from the interpolant gives
  $p_x \xor x=b$.  This is in contradiction with 
  $\lnot p_x \xor x=b$ from $\lnot C\proj B$.
  
  \else
  assume $\lnot C\proj B$ and $LA\W(-x_b+x_1,-\eps,\bot)$ hold:
  \[
  \exists x_1' \geq x_1.~ -x_b+x_1' \leq 0 \land (-x_b+x_1' \geq \eps \rightarrow \bot)
  \]
  Choose $x_1'$ for which the above formula holds. 
  From $\lnot C\proj B$ we have $-x_1+b < 0$ and
  $x_b=b$. Hence, $-x_b + x_1' \geq -b + x_1 > 0$.  This contradicts $-x_b+x_1'
  \leq 0$.
 
  Now consider the clause $C \equiv a= b\lor a< b \lor b< a$.
  We assume $\lnot C\proj B$ and $LA\W(x_1+x_2,0,x_b\neq x_1)$
  \[\exists x_1' \geq x_1, x_2' \geq x_2.~ 
  x_1'+x_2' \leq 0 \land (x_1'+x_2' \geq 0 \rightarrow x_b \neq x_1')\]
  and show a contradiction.  
  Choose $x_1'$ and $x_2'$ such that the formula holds.
  From $\lnot C\proj B$ we have $b\leq x_1$ and $-x_2\leq b$.  Thus
  \[0 \leq b - b \leq x_1+x_2 \leq x_1' + x_2' \leq 0,\]
  which gives $x_1'+x_2'=0$.  Also
  $b\leq x_1 \leq x_1' = -x_2' \leq -x_2 \leq b$, hence $x_b=b$. This
  contradicts $x_1'+x_2'\geq 0 \rightarrow x_b\neq b$.
  \fi
  \qed
\end{proof}

\end{techreport}

\subsection{Pivoting of Mixed Literals}

In this section we give the resolution rule for a step involving a
mixed inequality $a+b\leq c$ as pivot element.  In the following we
denote the auxiliary variable of the negated literal $\lnot(a+b\leq c)$
with $x_1$ and the
auxiliary variable of $a+b\leq c$ with $x_2$.  The intuition
here is that $x_1$ and $-x_2$ correspond to the same value between 
$a$ and $c-b$. The resolution rule for pivot element $a+b\leq c$ is
as follows where the values for $s_3$, $k_3$ and $F_3$ are given later.
\begin{equation} \label{rule:intla}\tag{rule-la}
\inferrule 
{C_1 \lor a+b\leq c : I_1[LA{(c_1x_1 + s_1(\vec x), k_1, F_1(x_1,\vec x))}] \\
C_2 \lor \lnot(a+b\leq c): I_2[LA{(c_2x_2 + s_2(\vec x), k_2, F_2(x_2, \vec x))}]} 
{C_1 \vee C_2: I_1[I_2[LA(s_3(\vec x), k_3, F_3(\vec x))]] } 
\end{equation}

\ifnewinterpolation The basic idea is to find for
 $\exists x_1. F_1(x_1, \vec x)\land F_2(-x_1, \vec x)$ an
equivalent quantifier free formula $F_3(\vec x)$.
To achieve this we note that we only have
to look on the value of $F_1$ for 
$-k_1 \leq c_1 x_1 + s_1(\vec x) \leq 0$,
since outside of this interval $F_1$ is guaranteed to be true
resp.\ false.  The formula $F_3$ must also be monotone and
satisfy the range condition.  We choose
\[ s_3(\vec x) = c_2 s_1(\vec x) + c_1s_2(\vec x), \]
and then $F_3$ will be false for $s_3(\vec x) > 0$, since either 
$F_1(x_1, \vec x)$ or $F_2(-x_1, \vec x)$ is false.
\else
The formula $LA(s_3(\vec x), k_3, F_3(\vec x))$ should
hold if and only if there is some $x_1=-x_2$ such that 
$LA(c_1x_1 + s_1(\vec x), k_1, F_1(x_1,\vec x))$ and 
$LA(c_2x_2 + s_2(\vec x), k_2, F_2(x_2,\vec x))$ hold.
From $c_1x_1 + s_1(\vec x) \leq 0$ and $c_2x_2+s_2(\vec x)\leq 0$ and $x_1=-x_2$ we get $c_2 s_1(\vec x) + c_1s_2(\vec x)\leq 0$, hence we choose
\[ s_3(\vec x) = c_2 s_1(\vec x) + c_1s_2(\vec x). \]
\fi
\ifnewinterpolation
The value of $k_3$ must be chosen such that $s_3(\vec x)<-k_3$ guarantees
the existence of a value $x_1$ with $c_1x_1 + s_1(\vec x)< -k_1$ and $-c_2x_1 + s_2(\vec x)< -k_2$.  Hence, in the integer case, the gap between
$\frac{s_2(\vec x)+k_2}{c_2}$ and $\frac{-s_1(\vec x)-k_1}{c_1}$
should be bigger than one.  Then,
$c_1c_2 < c_2(-s_1(\vec x)-k_1) - c_1(s_2(\vec x)+k_2)$.
So if we define
\[ k_3 = c_2k_1+c_1k_2+c_1c_2,\]
then there is a suitable $x_1$ for $s_3(\vec x) < - k_3$.
For $F_3$ we can then use a finite case distinction over all values where the truth value of $F_1$ is not determined.  This suggests defining
\begin{equation}
  \begin{array}{rl}
    F_3(\vec x) & :\equiv 
    \displaystyle{\bigvee_{i=0}^{\ceilfrac{k_1+1}{c_1}}
    F_1\left(\floorfrac{-s_1(\vec x)}{c_1} - i, \vec x\right)
    \land F_2\left(i-\floorfrac{-s_1(\vec x)}{c_1}, \vec x\right)}\\
  \end{array}
  \tag{int case}
\end{equation}
\else
For the inverse direction we need to guarantee the existence of 
$x_1=-x_2$ between $\frac{s_2(\vec x)}{c_2}$ and $\frac{-s_1(\vec x)}{c_1}$ 
such that  the following formulae hold\footnote{Unfortunately, the version published in TACAS 2013~\cite{tacaspaper} had this definition wrong.}:
\begin{align*}
LA_1^*(x_1) :\equiv s_1(\vec x) + c_1 x_1 \leq 0 \land (s_1(\vec x) + c_1 x_1 \geq -k_1 \rightarrow F_1(x_1,\vec x)),\\
LA_2^*(x_2) :\equiv s_2(\vec x) + c_2 x_2 \leq 0 \land (s_2(\vec x) + c_2 x_2 \geq -k_2 \rightarrow F_2(x_2,\vec x)).
\end{align*}

The first conjuncts of $LA_1^*$ and $LA_2^*$ yield
$\frac{s_2(\vec x)}{c_2}\leq -x_2$, $x_1 \leq \frac{-s_1(\vec x)}{c_1}$.  Hence
$x_1=-x_2$ should be a value between these bounds.  The second conjunct holds
vacuously if there is an integer value with
$\frac{s_2(\vec x)+k_2}{c_2}< -x_2 = x_1 < \frac{-s_1(\vec x)-k_1}{c_1}$. 
This holds if the gap is bigger than one:
$c_2s_1(\vec x) + c_1s_2(\vec x) < -c_2k_1 - c_1k_2-c_1c_2$.
Otherwise there are only finitely many candidates for $x_1=-x_2$ between 
$\frac{s_2(\vec x)}{c_2}$ and $\frac{-s_1(\vec x)}{c_1}$.  For these
we can do a finite case distinction in $F_3$.  This suggests the 
definitions
\begin{equation}
  \begin{array}{rl}
    k_3 &{}:= c_2k_1+c_1k_2+c_1c_2\\
    F_3(\vec x) & :\equiv 
    \displaystyle{\bigvee_{i=0}^{\ceilfrac{k_1+1}{c_1}}
    LA_1^*\left(\floorfrac{-s_1(\vec x)}{c_1} - i\right)
    \land LA_2^*\left(i-\floorfrac{-s_1(\vec x)}{c_1}\right)}\\
  \end{array}
  \tag{int case}
\end{equation}
\fi
\ifnewinterpolation 
In the real case, if $k_1 = -\eps$, the best choice is $x_1 = \frac{-s_1(\vec x)}{c_1}$, for which $F_1(x_1)$ is guaranteed to be true.  If $k_1 = 0$, we need to consider two cases:
\else
In the real case, we require that $k_1$ and $k_2$ are either $-\eps$ or $0$.
Then, the only candidate for $x_1$ is $\frac{-s_1(\vec x)}{c_1}$.  We define
\fi
\begin{equation}
  \ifnewinterpolation 
  \begin{array}{rl}
    k_3 &{}:= \begin{cases}k_2  & \text{if }k_1=-\eps\\
                           0 & \text{if }k_1=0
             \end{cases}\\
    F_3(\vec x) &{}:= \begin{cases}
       F_2\left(\frac{s_1(\vec x)}{c_1}, \vec x\right) & \text{if }k_1=-\eps\\
       s_3(\vec x) < 0 \lor 
       \left(F_1\left(-\frac{s_1(\vec x)}{c_1}, \vec x\right) \land
       F_2\left(\frac{s_1(\vec x)}{c_1}, \vec x\right)\right)
       & \text{if }k_1=0
             \end{cases}
  \end{array}
  \else
  \begin{array}{rl}
    k_3 &{}:= \begin{cases}-\eps & \text{if }k_1=k_2=-\eps\\
                          0     & \text{if }k_1=0 \lor k_2 = 0
             \end{cases}\\
    F_3(\vec x) & :\equiv
    LA_1^*\left(\frac{-s_1(\vec x)}{c_1}\right) \land 
    LA_2^*\left(-\frac{-s_1(\vec x)}{c_1}\right)
  \end{array}
  \fi
  \tag{real case}
\end{equation}

\begin{techreport}

Note that the formula of the integer case is asymmetric.  If
$\ceilfrac{k_2+1}{c_2} < \ceilfrac{k_1+1}{c_1}$ we can replace $-s_1$
by $s_2$, $k_1$ by $k_2$, and $c_1$ by $c_2$.  This leads to a fewer 
number of disjuncts in $F_3$.
\ifnewinterpolation  
Also note that we can remove $F_1$ from the last disjunct of $F_3$, 
as it will always be true. \fi

\end{techreport}

With these definitions we can state the following lemma.
\ifnewinterpolation
\begin{lemma}\label{lemma_la}
  Let for $i=1,2$, $s_i(\vec x)$ be linear terms over $\vec x$,
  $c_i \geq 0$, $k_i \in\mathbb{Z}_{\geq -1}$
  (integer case) or $k_i\in\{0,-\eps\}$ (real case), 
  $F_i(x_i,\vec x)$ monotone formulas with
  $F_i(x_i,\vec x)= \bot$ for $c_i x_i + s_i(\vec x) > 0$ and 
  $F_i(x_i,\vec x)= \top$ for $c_i x_i + s_i(\vec x) < -k_i$.  
  Let $s_3,k_3, F_3$ be as defined above.
  Then $F_3$ is monotone, $F_3(\vec x)=\bot$ for $s_3(\vec x) > 0$ and
  $F_3(\vec x)=\top$ for $s_3(\vec x)< -k_i$.
\end{lemma}
\begin{proof}
  Since $F_1$ and $F_2$ are monotone and they occur only positively in
  $F_3$, $F_3$ must also be monotone.  If $s_3(\vec x)> 0$, then
  $\frac{-s1(\vec x)}{c_1} < \frac{s_2}{c_2}$.  Hence, for 
  every $x \leq \frac{-s(\vec x)}{c_1}$, $F_2(-x, \vec x)$ is false 
  since $-c_2 x + s_2(\vec x) > 0$.  By definition, every disjunct 
  of $F_3$ (except $s_3(\vec x) < 0$) 
  contains $F_2(-x,\vec x)$ for such an $x$, so $F_3(\vec x)$ 
  is false.

  Now assume $s_3(\vec x) < -k_3$. For $k_1=-\eps$ in the real case,
  $F_3(\vec x) =
  F_2(-\frac{s_1(\vec x)}{c_1})$ is true since $s_1(\vec x) + s_2(\vec
  x) < -k_2$.  For $k_1=0$, $F_3$ is true by definition.
  In the integer case define $y := \floorfrac{-s_1(\vec x)}{c_1} -
  \ceilfrac{k_1+1}{c_1}$.  
  This implies $c_1 y \leq -s_1(\vec x)-k_1-1$, hence $F_1(y, \vec x)$ holds.
  Also $c_1 y \geq -s_1(\vec x)-k_1-c_1$, hence
  \[c_1c_2y + c_1s_2(\vec x) \geq -s_3(\vec x) - c_2k_1 - c_1c_2 > k_3 - c_2k_1 - c_1c_2 =  c_1k_2.\]
  Therefore, $F_2(-y,\vec x)$ holds. Since $y$ is included in the big
  disjunction of $F_3$, $F_3(\vec x)$ is true.
\end{proof}
\fi

\begin{lemma}\label{lemma_la}
  Let for $i=1,2$, $s_i(\vec x)$ be linear terms over $\vec x$,
  $c_i \geq 0$, $k_i \in\mathbb{Z}_{\geq -1}$
  (integer case) or $k_i\in\{0,-\eps\}$ (real case), 
  \ifnewinterpolation
  $F_i(x_i,\vec x)$ monotone formulas with
  $F_i(x_i,\vec x)= \bot$ for $c_i x_i + s_i(\vec x) > 0$ and 
  $F_i(x_i,\vec x)= \top$ for $c_i x_i + s_i(\vec x) < -k_i$.  
  Let $\ifnewinterpolation\else s_3,k_3,\fi F_3$ be as defined above.
  \else
  $F_i(x_i,\vec x)$ arbitrary formulae and $s_3,k_3, F_3$ as defined above.
  \fi
  Then 
  \ifnewinterpolation
  \[
    F_3(\vec x) \leftrightarrow (\exists x_1. F_1(x_1, \vec x)\land F_2(-x_1, \vec x))
    \]
  \else
  \[
    \begin{array}{c}
    (\exists x_1. LA\S(c_1 x_1 + s_1(\vec x), k_1, F_1(x_1, \vec x))\land
                  LA\S(-c_2 x_1 + s_2(\vec x), k_2, F_2(-x_1, \vec x)))\\
    \rightarrow  LA\S(s_3(\vec x),k_3, F_3(\vec x))
    \end{array}
  \]
  and
  \[
    \begin{array}{c}
    LA\W(s_3(\vec x),k_3, F_3(\vec x)) \rightarrow \\
    (\exists x_1. LA\W(c_1 x_1 + s_1(\vec x), k_1, F_1(x_1, \vec x))\land
    LA\W(-c_2 x_1 + s_2(\vec x), k_2, F_2(-x_1, \vec x)))
    \end{array}
  \]
  \fi
\end{lemma}

\begin{techreport}
\ifnewinterpolation
\begin{proof}[for \laz]
Since $F_3$ is a disjunction of $F_1(x, \vec x) \land F_2(-x, \vec x)$
for different values of $x$, the implication from left to right is obvious.
We only need to show the other direction.  
For this, choose $x_1$ such that $F_1(x_1,\vec x) \land F_2(-x_1, \vec x)$ 
holds. We show $F_3(\vec x)$.
We define $y := \floorfrac{-s_1(\vec x)}{c_1} -
\ceilfrac{k_1+1}{c_1}$.  This implies $y \leq \frac{-s_1(\vec x)-k_1-1}{c_1}$.
We show $F_3$ by a case split on $x_1 < y$.

\paragraph{Case $x_1 < y$.}
Since $F_2$ is monotone and $-x_1 > -y$, we have $F_2(-y,\vec x)$.  Also
$F_1(y,\vec x)$ holds since $c_1 y + s_1(\vec x) < -k_1$.  This 
implies $F_3(\vec x)$,
since $F_1(y,\vec x)\land F_2(-y,\vec x)$ is a disjunct of $F_3$.

\paragraph{Case $y \leq x_1$.}
Since $F_1(x_1,\vec x)$ holds, $c_1 x_1 + s_1(\vec x) \leq 0$, hence
$x_1 \leq \floorfrac{-s_1(\vec x)}{c_1}$.  Thus, $x_1$ is one of the values
$\floorfrac{-s_1(\vec x)}{c_1} - i$ for $0 \leq i \leq \ceilfrac{k_1+1}{c_1}$.
This means the disjunction $F_3(\vec x)$ includes $F_1(x_1,\vec x)\land F_2(-x_1,\vec x)$.\qed
\end{proof}

\begin{proof}[for \laq]
In the case $k_1 = -\eps$, $F_1(\frac{-s_1}{c_1},\vec x)$ is true.
From the definition of $F_3$, we get the implication $F_3(\vec x)
\rightarrow \exists x_1. F_1(x_1,\vec x) \land F_2(-x_1,\vec x)$ for
$x_1= \frac{-s_1}{c_1}$.  If $k_1 = 0$ and $s_3(x) < 0$, then
$\frac{s_2}{c_2} < \frac{-s_1}{c_1}$ and for any value $x_1$ in between,
$F_1(x_1,\vec x) \land F_2(-x_1,\vec x)$ are true.

For the other direction assume that $F_1(x_1, \vec x) \land F_2(-x_1, \vec x)$
holds.  Since $F_1$ is not false, $x_1\leq \frac{-s_1}{c_1}$ holds.
If $x_1 = \frac{-s_1}{c_1}$ then $F_3$ holds by definition.
In the case $k_1 = 0$ where $x_1 < \frac{-s_1}{c_1}$, we have 
$s_3(\vec x) < 0$, since $F_2(-x_1, \vec x)$ is not false.  
In the case $k_1 = -\eps$, we need to show that $F_2(\frac{s_1}{c_1},\vec x)$
holds. This follows from $x_1\leq \frac{-s_1}{c_1}$ and monotonicity of $F_2$.
\qed
\end{proof}
\fi

\end{techreport}

This lemma can be used to show that (\ref{rule:intla}) is correct.
\begin{theorem}[Soundness of (\ref{rule:intla})]
  Let $a+b\leq c$ be a mixed literal with the auxiliary variable $x_2$, and
  $x_1$ be the auxiliary variable of the negated literal.  If
  $I_1[LA(c_1 x_1 + s_1, k_1, F_1)]$ 
  \ifnewinterpolation is a partial interpolant 
  \else yields two partial interpolants (strong and weak) \fi
  of $C_1 \lor a+b\leq c$ and 
  $I_2[LA(c_2 x_2 + s_2, k_2, F_2)]$ 
  \ifnewinterpolation is a partial interpolant 
  \else yields two partial interpolants \fi
  of $C_2 \lor \lnot(a+b\leq c)$ then
  $I_1[I_2[LA(s_3, k_3, F_3)]]$
  \ifnewinterpolation is a partial interpolant 
  \else yields two partial interpolants \fi
  of the clause $C_1 \vee C_2$.
\end{theorem}
%\begin{tacas}
% The proofs for the lemma and the theorem are given in the technical report~\cite{atr}.
%\end{tacas}

To ease the presentation, we gave the rule (\ref{rule:intla}) with only one
$LA$ term per partial interpolant. The generalised rule requires the partial interpolants of
the premises to have the shapes $I_1[LA_{11}]\dots[LA_{1n}]$ and
$I_2[LA_{21}]\dots[LA_{2m}]$.  The resulting interpolant is
\[I_1[I_2[LA_{311}]\dots[LA_{31m}]]\dots [I_2[LA_{3n1}]\dots[LA_{3nm}]]\]
where $LA_{3ij}$ is computed from $LA_{1i}$ and $LA_{2j}$
as explained above.

\begin{techreport}
\begin{proof}
\ifnewinterpolation\else
We already showed that the strong interpolant implies the weak
interpolant for the interpolation pattern used.
\fi
The symbol condition
holds for $I_3$ if it holds for $I_1$ and $I_2$, which can be seen as
follows.  The only symbol that is allowed to occur in $I_1$ resp.\
$I_2$ but not in $I_3$ is the auxiliary variable introduced by the
literal, i.e., $x_1$ resp.\ $x_2$.  This variable may only occur
inside the $LA_1$ resp.\ $LA_2$ terms as indicated and, by
construction, $x_1$ and $x_2$ do not occur in $LA_3$.  Furthermore,
the remaining variables from $\vec x$ occur in $s_3(\vec x)$ with a
positive coefficient as required by our pattern and occur only inside
the $LA$ pattern in $s_3$ and $F_3$.  Thus $I_3$ has the required
form.  We will use Lemma~\ref{lem:weakstrongip} to show that $I_3$ is
a partial interpolant.  For this we need to show inductivity~(ind) and
contradiction~(cont).

In this proof we will use $I_1[LA_{1i}(x_1)]$ to denote the first interpolant
\[I_1[LA(s_{11}+c_{11}x_1, k_{11}, F_{11})]\dots[LA(s_{1n}+c_{1n}x_1, k_{1n}, F_{1n})]\]
and similarly $I_2[LA_{2j}(x_2)]$ and $I_1[I_2[LA_{3ij}]]$, the latter standing for
\[I_1[I_2[LA_{311}]\dots[LA_{31m}]]\dots [I_2[LA_{3n1}]\dots[LA_{3nm}]]\]
where 
\[
LA_{3ij} = LA(c_{2j}s_{1i}+c_{1i}s_{2j}, k_{3ij}, F_{3ij}).
\]

\subsubsection*{Inductivity.}
We apply \ifnewinterpolation\else the first part of \fi Lemma~\ref{lemma_la} on $x_1=a$, which
gives us
\[\bigwedge_{ij} LA_{1i}\S(a) \wedge LA_{2j}\S(-a) \rightarrow LA_{3ij}\S\]
Using the deep substitution lemma, we obtain 
\[
  I_{1}\S\left[LA_{1i}\S(a)\right]\wedge I_{2}\S\left[LA_{2j}\S(-a)\right]\rightarrow I_{1}\S\left[I_{2}\S\left[LA_{3ij}\S\right]\right].
\tag{$*$}
\]

Now assume the left-hand-side of (ind), which in this case is
\[
    \forall x_1, x_2.\:  
      (-a + x_1 \leq 0\rightarrow I_1\S[LA\S_{1i}(x_1)])  \land 
      (a + x_2 \leq 0 \rightarrow I_2\S[LA\S_{2j}(x_2)]).
\]
Instantiating $x_1$ with $a$ and $x_2$ with $-a$ gives us
$I_1\S[LA\S_{1i}(a)]$ and $I_2\S[LA\S_{2j}(-a)]$.  
Thus by $(*)$, $I_{3}\S\equiv I_1\S[I_2\S[LA_{3ij}\S]]$ holds as desired.

\subsubsection*{Contradiction.}
We assume $I_1\W[I_2\W[LA_{3ij}\W]]$ and show
\[
  \exists x_1,x_2.\:
  (-x_1 - b < -c \land I_1\W[LA_{1i}\W(x_1)])  \lor
  (-x_2 +  b \leq c \land I_2\W[LA_{2j}\W(x_2)]) \tag{$*$}
\]
We do a case distinction on
\[\bigwedge_i (I_2\W[LA_{3ij}\W] \rightarrow 
 \exists x_1.\: x_1 > c - b \land LA_{1i}\W(x_1))\]
If it holds, then we may get a different value for $x_1$ for every
$i$.  However, if $LA_{1i}\W(x_1)$ holds for some value, it also holds
for any smaller value of $x_1$.  Take $x$ as the minimum of these
values (or $x=c-b+1$ if the implication holds vacuously for every
$i$).  Then, $-x -b < -c$ and
$\bigwedge_i (I_2\W[LA_{3ij}\W] \rightarrow LA_{1i}\W(x))$.
With monotonicity we get from $I_1\W[I_2\W[LA_{3ij}\W]]$ that
$I_1\W[LA_{1i}\W(x)]$ holds. Hence, the left disjunct of formula $(*)$ holds.

In the other case there is some $i$ with
\[I_2\W[LA_{3ij}\W] \land (\forall x_1.\: x_1 > c - b \rightarrow \lnot LA_{1i}\W(x_1)). \tag{$**$}\]
The second part of Lemma~\ref{lemma_la} gives us 
\[\bigwedge_j( LA_{3ij}\W \rightarrow \exists x_1.\: LA_{1i}\W(x_1) \land LA_{2j}\W(-x_1))\]
Then, $x_1 \leq c-b$ by $(**)$.  But if $LA_{2j}\W(-x_1)$ holds, then
$LA_{2j}\W$ also holds for the smaller value $b-c$.  This gives us
\[ \bigwedge_j (LA_{3ij}\W \rightarrow LA_{2j}\W(b-c))\]
We obtain $I_2\W[LA_{2j}\W(b-c)]$ by applying monotonicity on the
left conjunct of formula $(**)$.  Thus the
right disjunct of formula $(*)$ holds for $x_2 = b-c$.
\qed
\end{proof}
\end{techreport}

%%  LocalWords:  unsatisfiability Farkas DPLL interpolants versa
%%  LocalWords:  interpolant Oppen disequality

\def\vt{\vphantom{t}}

\section{An Example for the Combined Theory}
The previous examples showed how to use our technique to compute an
interpolant in the theory of uninterpreted functions, or the theory of linear
arithmetic. We will now present an example in the combination of these
theories by applying our scheme to a proof of unsatisfiability of the
interpolation problem
\begin{align*}
  A &\equiv t\leq 2a \land 2a \leq s \land f(a)=q\\
  B &\equiv s\leq 2b \land 2b \leq t+1 \land \lnot(f(b)= q)
\end{align*}
where $a$, $b$, $s$, and $t$ are integer constants, $q$ is a constant of the
uninterpreted sort $U$, and $f$ is a function from integer to $U$.

We derive the interpolant using Pudl\'ak's algorithm and the rules shown in this
paper. Note that the formula is already in conjunctive normal form. Since we
use Pudl\'ak's algorithm, every input clause is labelled with $\bot$ if it is
an input clause from $A$, and $\top$ if it is an input clause from $B$.  We
will simplify the interpolants by removing neutral elements of Boolean
connectives.

Since the variables $a$ and $b$ are shared between the theory of uninterpreted
functions and the theory of linear arithmetic, we get some theory combination
clauses for $a$ and $b$.  The only theory combination clause needed to prove
unsatisfiability of $A\land B$ is $a=b\lor\lnot(b\leq a)\lor\lnot(a\leq b)$
which has the partial interpolant $LA(x_1+x_2,0,\ifnewinterpolation F[EQ(x,x_1)]\else EQ(x,x_1)\fi)$\ifnewinterpolation{} where $F[G] \equiv x_1\leq -x_2 \land (x_1\geq -x_2 \rightarrow G)$\fi.  Here, $x_1$ is
used to purify\footnote{Note that we purify the conflict, i.\,e., the negated
  clause} $b\leq a$ and $x_2$ is used to purify $a\leq b$.

We get two lemmas from \laz: The first one, $\lnot(2a\leq s)\lor \lnot(s\leq
2b)\lor a\leq b$, states that we can derive $a\leq b$ from $2a\leq s$ and
$s\leq 2b$. Let $x_3$ be the variable used to purify $\lnot(a\leq b)$.  Note
that we purify the literals in the conflict, i.\,e., the negation of the
lemma. Then, this lemma can be annotated with the partial interpolant
$LA(2x_3-s,-1,\ifnewinterpolation 2x_3 \leq s\else \bot\fi)$.  We can resolve this lemma with the unit clauses from
the input to get $a\leq b$.
\begin{gather*}
\inferrule*
{ 
  \inferrule*[rightskip=\ifnewinterpolation2cm \else .5cm \fi]
  { \lnot(2a\leq s)\lor \lnot(s\leq 2b)\lor a\leq b : LA(2x_3-s,-1,\ifnewinterpolation 2x_3 \leq s\else \bot\fi) \\
    2a\leq s : \bot }
  { \lnot(s\leq 2b)\lor a\leq b : LA(2x_3-s,-1,\ifnewinterpolation 2x_3 \leq s\else \bot\fi) }\\
  s\leq 2b : \top }
{ a\leq b : LA(2x_3-s,-1,\ifnewinterpolation 2x_3 \leq s\else\bot\fi) }
\end{gather*}

The second \laz-lemma, $\lnot(t\leq 2a)\lor\lnot(2b\leq t+1)\lor b\leq a$,
states that we can derive $b\leq a$ from $t\leq 2a$ and $2b\leq t+1$. Let
$x_4$ be the variable used to purify $\lnot(b\leq a)$. Then, we can annotate
the lemma with the partial interpolant $LA(2x_4+t,-1,\ifnewinterpolation 2x_4 + t \leq 0\else\bot\fi)$ and propagate this
partial interpolant to the unit clause $b\leq a$ by resolution with input
clauses.
\begin{gather*}
\inferrule*{
  \inferrule*[rightskip=\ifnewinterpolation 2.5cm \else 1.5cm\fi]
  { \lnot(t\leq 2a)\lor\lnot(2b\leq t+1)\lor b\leq a : LA(2x_4+t,-1,\ifnewinterpolation 2x_4+t \leq 0\else\bot\fi)\quad
    t\leq 2a : \bot }
  { \lnot(2b\leq t+1)\lor b\leq a : LA(2x_4+t,-1,\ifnewinterpolation 2x_4+t \leq 0\else\bot\fi) } \\
   2b\leq t+1 : \top }
{ b\leq a : LA(2x_4+t,-1,\ifnewinterpolation 2x_4+t \leq 0\else\bot\fi) }
\end{gather*}

Additionally, we get one lemma from \euf,
$f(b)=q\lor\lnot(f(a)=q)\lor\lnot(a=b)$, that states that, given $f(a)=q$ and
$a=b$, by congruence, $f(b)=q$ has to hold. Let $x$ be the variable used to
purify $a=b$. Then, we can label this lemma with the partial interpolant
$f(x)=q$. Note that this interpolant has the form $I(x)$ as required by our
interpolation scheme. We propagate this partial interpolant to the unit clause
$\lnot(a=b)$ by resolving the lemma with the input clauses.
\[
\inferrule*
{ 
  \inferrule*[rightskip=.5cm]
  { f(b)=q\lor\lnot(f(a)=q)\lor\lnot(a=b) : f(x)=q \\ f(b)=q : \top}
  { \lnot(f(a)=q)\lor\lnot(a=b) : f(x)=q}\\
  f(a)=q : \bot}
{ \lnot(a=b) : f(x) = q }
\]

From the theory combination clause $a=b\lor
\lnot(b\leq a)\lor\lnot(a\leq b)$ 
and the three unit clauses derived above, we show a contradiction.  
We start by resolving with the unit clause $a=b$ using (\ref{rule:inteq}) and 
produce the partial interpolant $LA(x_1+x_2,0,f(x_1)=q)$.
\[
\inferrule*
{ a=b\lor\lnot(b\leq a)\lor\lnot(a\leq b) : LA(x_1+x_2,0,\ifnewinterpolation F[EQ(x,x_1)] \else EQ(x,x_1)\fi) \\ 
  \lnot(a=b) : f(x) = q}
{ \lnot(b\leq a)\lor\lnot(a\leq b) : LA(x_1+x_2,0,
\ifnewinterpolation F[f(x_1) = q] \else f(x_1) = q\fi) }
\]

The next step resolves on $b\leq a$ using (\ref{rule:intla}).  Note that we
used $x_1$ to purify $b\leq a$ and $x_4$ to purify $\lnot(b\leq a)$.  Hence,
these variables will be removed from the resulting partial interpolant. From
the partial interpolants of the antecedents, $LA(2x_4+t,-1,\ifnewinterpolation 2 x_4 + t \leq 0 \else \bot \fi)$ and
$LA(x_1+x_2,0,\ifnewinterpolation F[f(x_1)=q] \else f(x_1) = q\fi)$, we get the following components:
\begin{align*}
  c_1&=2
  &s_1&=t
  &k_1&=-1
  &F_1(x_4)&\equiv \ifnewinterpolation 2 x_4 + t \leq 0 \else \bot \fi \\
  c_2&=1
  &s_2&=x_2
  &k_2&=0
  &F_2(x_1)&\equiv \ifnewinterpolation F[f(x_1) = q] \else f(x_1) = q\fi
\end{align*}

These components yield $k_3=1\cdot (-1)+2\cdot 0+2\cdot 1=1$.  Furthermore,
$\ceilfrac{k_1+1}{c_1} = 0$ leads to one disjunct in $F_3$. The corresponding
values are $\floorfrac{-t}{2}$, resp.\ $-\floorfrac{-t}{2}$.
\ifnewinterpolation $F_1(\floorfrac{-t}{2})$ is always true and can be omitted.
\fi
The resulting formula $G(x_2) := F_3(\vec x)$ is
\ifnewinterpolation
\begin{align*}
  G(x_2)\equiv{} & -\floorfrac{-t}{2} \leq -x_2\land \left(\floorfrac{-t}{2}\geq -x_2\rightarrow f\left(-\floorfrac{-t}{2}\right)=q\right).
\end{align*}
\else
\begin{align*}
  G(x_2)\equiv{} &t+2\floorfrac{-t}{2}\leq 0\land\left(t+2\floorfrac{-t}{2}\geq 1\rightarrow \bot\right)\land{}\\
  &x_2-\floorfrac{-t}{2}\leq 0\land\left(x_2-\floorfrac{-t}{2}\geq 0\rightarrow f\left(-\floorfrac{-t}{2}\right)=q\right)\\
  {}\equiv{} & x_2-\floorfrac{-t}{2}\leq 0\land \left(x_2-\floorfrac{-t}{2}\geq 0\rightarrow f\left(-\floorfrac{-t}{2}\right)=q\right).
\end{align*}
Note that the first two conjuncts simplify to $\top$ and we remove
them.
\fi
The partial interpolant for the clause $\lnot(a\leq b)$ is
$LA(t+2x_2,1,G(x_2))$.
\[
\inferrule*{
  b\leq a : LA(2x_4+t,-1,\ifnewinterpolation 2x_4 + t \leq 0 \else \bot\fi) \\
  \lnot(b\leq a)\lor\lnot(a\leq b) : LA(x_1+x_2,0,f(x_1) = q) }
{
  \lnot(a\leq b) : LA(t+2x_2,1,G(x_2))
}
\]

In the final resolution step, we resolve $a\leq b$ labelled with partial
interpolant $LA(2x_3-s,-1,\ifnewinterpolation 2 x_3 \leq s \else \bot\fi)$ against $\lnot(a\leq b)$ labelled with
$LA(t+2x_2,1,G(x_2))$. Note that the literals have been purified with $x_3$ and
$x_2$, respectively. We get the components
\begin{align*}
  c_1&=2
  &s_1&=-s
  &k_1&=-1
  &F_1(x_3)&\equiv\ifnewinterpolation 2 x_3 \leq s \else \bot\fi\\
  c_2&=2
  &s_2&=t
  &k_2&=1
  &F_2(x_2)&\equiv G(x_2).
\end{align*}

We get $k_3=2\cdot (-1)+2\cdot 1 + 2\cdot 2=4$.  Again,
$\ceilfrac{k_1+1}{c_1}=0$ yields one disjunct in $F_3$ with the values
$\floorfrac{s}{2}$, and $-\floorfrac{s}{2}$, respectively.  
\ifnewinterpolation Again, $F_1(\floorfrac{s}{2})$ is always true and can be omitted.
\fi
The resulting
formula is
\ifnewinterpolation
\begin{align*}
H\equiv{}& G\left(-\floorfrac{s\vt}{2}\right)\\
{}\equiv{}& -\floorfrac{-t}{2} \leq \floorfrac{s\vt}{2}\land \left(\floorfrac{-t}{2}\geq \floorfrac{s\vt}{2}\rightarrow f\left(-\floorfrac{-t}{2}\right)=q\right).
\end{align*}
\else
\begin{align*}
H\equiv{}&{-s}+2\floorfrac{s\vt}{2}\leq 0\land\left(-s+2\floorfrac{s\vt}{2}\geq 1 \rightarrow\bot\right)\land{}\\
&t-2\floorfrac{s\vt}{2}\leq 0\land\left(t-2\floorfrac{s\vt}{2}\geq -1 \rightarrow G\left(-\floorfrac{s\vt}{2}\right)\right)\\
{}\equiv{}&t-2\floorfrac{s\vt}{2}\leq 0\land\left(t-2\floorfrac{s\vt}{2}\geq -1 \rightarrow -\floorfrac{s\vt}{2}-\floorfrac{-t}{2}\leq 0 \land{} \right.\\
&\quad\left.\left(-\floorfrac{s\vt}{2}-\floorfrac{-t}{2}\geq 0\rightarrow f\left(-\floorfrac{-t}{2}\right)=q\right)\right).
\end{align*}
Again, the first two conjuncts trivially simplify to $\top$ and can be removed.
\fi

The final resolution step yields an interpolant for this problem.
\[
\inferrule{
  a\leq b : LA(2x_3-s,-1,\bot) \\
  \lnot(a\leq b) : LA(t+2x_2,1,G(x_2))
}
{
  \bot : LA(-2s+2t,4,H)
}
\]

\ifnewinterpolation
Thus $H$ is the final interpolant.
\else
We obtain the final interpolant by unfolding the $LA$-form
and some simplifications:
\[
  t\leq s \land \left(t \geq s-2\rightarrow
  \left(t\leq 2\floorfrac{s\vt}{2}\land\left(-\floorfrac{-t}{2}\geq \floorfrac{s\vt}{2}\rightarrow
  f\left(-\floorfrac{-t}{2}\right)=q\right)\right)\right).
\]
\fi
%or with more simplifications and using the identity $\lceil x\rceil = -\floor{-x}$:
%\[
%  t\leq s \land \left(t \geq s-2\rightarrow
%  \left(\frac{t}{2} \leq \floorfrac{s\vt}{2}\land\left(\ceilfrac{t}{2}\geq \floorfrac{s\vt}{2}\rightarrow
%  f\left(\ceilfrac{t}{2}\right)=q\right)\right)\right).
%\]
% 
Now we argue validity of this interpolant.
\paragraph{Interpolant follows from the $A$-part.}
\ifnewinterpolation
The $A$-part contains $2a\leq s$, which implies $a \leq \floorfrac{s}{2}$.
From $t\leq 2a$ we get $-\floorfrac{-t}{2}\leq a$.  Hence,
$-\floorfrac{-t}{2} \leq\floorfrac{s}{2}$.
Moreover, $-\floorfrac{-t}{2}\geq \floorfrac{s}{2}$ implies
$-\floorfrac{-t}{2}=a$.
So with the $A$-part we get $f(-\floorfrac{-t}{2})=q$.
\else
The interpolant expresses that $t\leq s$ holds, which can be deduced from 
the $A$-part.
Moreover from $2a\leq s$, we get $a \leq \floorfrac{s}{2}$.
With $t\leq 2a$ we get $t\leq 2\floorfrac{s}{2}$.  Finally, we show 
$-\floorfrac{-t}{2}\geq \floorfrac{s}{2} \rightarrow f(-\floorfrac{-t}{2})=q$.
Using $-2a \leq -t$, we get $-a \leq \floorfrac{-t}{2}$.
Hence, $a \leq \floorfrac{s}{2} \leq -\floorfrac{-t}{2} \leq a$ implies 
$a=-\floorfrac{-t}{2}$, so with the $A$-part $f(-\floorfrac{-t}{2})=q$
follows.
\fi

\paragraph{Interpolant is inconsistent with the $B$-part.}
The $B$-part implies $s\leq 2b \leq t+1$.  
Hence, we have $\floorfrac{s}{2} \leq b\leq \floorfrac{t+1}{2}$.
A case distinction on whether $t$ is even or odd yields
 $\floorfrac{t+1}{2} = -\floorfrac{-t}{2}$.  Therefore,
$\floorfrac{s}{2} \leq b\leq -\floorfrac{-t}{2}$ holds.
Hence, the interpolant guarantees $f(-\floorfrac{-t}{2})=q$ and
\ifnewinterpolation
$-\floorfrac{-t}{2} \leq \floorfrac{s}{2}$. 
\else
$t\leq 2\floorfrac{s}{2}$. The latter implies
$-\floorfrac{-t}{2} \leq \floorfrac{s}{2}$. 
\fi
Hence, $b=-\floorfrac{-t}{2}$
and with $f(b)\neq q$ from the $B$-part we get a contradiction.

\paragraph{Symbol condition is satisfied.}
The symbol condition is trivially satisfied since $\symb(A)=\{a,t,s,f,q\}$ and
$\symb(B)=\{b,t,s,f,q\}$.  The shared symbols are $t$, $s$, $f$, and $q$ which
are exactly the symbols occurring in the interpolant.

%Since the interpolant follows from the $A$-part, is inconsistent with the
%$B$-part, and only contains shared symbols, it is a valid interpolant for the
%interpolation problem $A$ and $B$.

\ifnewinterpolation\else
A close inspection of the last proof reveals that $H$ is already a valid
interpolant of $A$ and $B$. 
This shows that in a certain sense the produced interpolants are not minimal.
It may be useful to investigate more closely which parts can be safely
omitted.
\fi

%----------------------------------------
\section{Conclusion and Future Work}
%----------------------------------------

We presented a novel interpolation scheme to extract Craig interpolants from
resolution proofs produced by SMT solvers without restricting the solver or
reordering the proofs.  The key ingredients of our method are virtual
purifications of troublesome mixed literals, syntactical restrictions of
partial interpolants, and specialised interpolation rules for pivoting steps
on mixed literals.

In contrast to previous work, our interpolation scheme does not need
specialised rules to deal with extended branches as commonly used in
state-of-the-art SMT solvers to solve \laz-formulae.  Furthermore, our scheme
can deal with resolution steps where a mixed literal occurs in both
antecedents, which are forbidden by other schemes~\cite{Cimatti2010,Goel2009}.

Our scheme works for resolution based proofs in the DPLL(T) context provided
there is a procedure that generates partial interpolants with our syntactic
restrictions for the theory lemmas.  We sketched these procedures for the
theory lemmas generated by either congruence closure or linear arithmetic solvers
producing Farkas proofs. 
In this paper, we limited the presentation to the combination of the theory of
uninterpreted functions, and the theory of linear arithmetic over the integers
or the reals.  Nevertheless, the scheme could be extended to support other
theories.  This requires defining the projection functions for mixed
literals in the theory, defining a pattern for 
\ifnewinterpolation\else weak and strong \fi partial
interpolants, and proving a corresponding resolution rule.

We plan to produce interpolants of different strengths using the
technique from D'Silva et al.~\cite{D'Silva2010}.  This is orthogonal to our
interpolation scheme (particularly to the 
\ifnewinterpolation partial \else weak and strong \fi
interpolants used for mixed literals).
Furthermore, we want to extend the correctness proof to show that our scheme 
works with inductive sequences of interpolants~\cite{mcmillan06lai} and tree
interpolants~\cite{HHP10}. We also plan to extend this scheme to other theories
including arrays and quantifiers.

\bibliographystyle{plain}
\bibliography{mixedInterpolation-techReport}

\begin{thebibliography}{10}

\bibitem{DBLP:conf/cav/BeyerZM08}
Dirk Beyer, Damien Zufferey, and Rupak Majumdar.
\newblock {CSIsat}: Interpolation for {LA+EUF}.
\newblock In {\em CAV}, pages 304--308. Springer, 2008.

\bibitem{DBLP:conf/vmcai/BrilloutKRW11}
Angelo Brillout, Daniel Kroening, Philipp R{\"u}mmer, and Thomas Wahl.
\newblock Beyond quantifier-free interpolation in extensions of {Presburger}
  arithmetic.
\newblock In {\em VMCAI}, pages 88--102. Springer, 2011.

\bibitem{Bruttomesso2010}
Roberto Bruttomesso, Simone Rollini, Natasha Sharygina, and Aliaksei Tsitovich.
\newblock Flexible interpolation with local proof transformations.
\newblock In {\em ICCAD}, pages 770--777. IEEE, 2010.

\bibitem{toolpaper}
J{\"u}rgen Christ, Jochen Hoenicke, and Alexander Nutz.
\newblock {SMTInterpol}: An interpolating {SMT} solver.
\newblock In {\em SPIN}, pages 248--254. Springer, 2012.

\bibitem{Cimatti2010}
Alessandro Cimatti, Alberto Griggio, and Roberto Sebastiani.
\newblock Efficient interpolant generation in satisfiability modulo theories.
\newblock In {\em TACAS}, pages 397--412. Springer, 2008.

\bibitem{Craig2011}
William Craig.
\newblock Three uses of the {Herbrand-Gentzen} theorem in relating model theory
  and proof theory.
\newblock {\em J. Symb. Log.}, 22(3):269--285, 1957.

\bibitem{Detlefs2005}
David Detlefs, Greg Nelson, and James~B. Saxe.
\newblock Simplify: {A} theorem prover for program checking.
\newblock {\em J. ACM}, 52(3):365--473, 2005.

\bibitem{Dillig2011}
Isil Dillig, Thomas Dillig, and Alex Aiken.
\newblock Cuts from proofs: A complete and practical technique for solving
  linear inequalities over integers.
\newblock In {\em CAV}, pages 233--247. Springer, 2009.

\bibitem{D'Silva2010}
Vijay D'Silva, Daniel Kroening, Mitra Purandare, and Georg Weissenbacher.
\newblock Interpolant strength.
\newblock In {\em VMCAI}, pages 129--145. Springer, 2010.

\bibitem{Dutertre2006}
Bruno Dutertre and Leonardo de~Moura.
\newblock A fast linear-arithmetic solver for {DPLL(T)}.
\newblock In {\em CAV}, pages 81--94. Springer, 2006.

\bibitem{Fuchs2009}
Alexander Fuchs, Amit Goel, Jim Grundy, Sava Krstic, and Cesare Tinelli.
\newblock Ground interpolation for the theory of equality.
\newblock In {\em TACAS}, pages 413--427. Springer, 2009.

\bibitem{Goel2009}
Amit Goel, Sava Krstic, and Cesare Tinelli.
\newblock Ground interpolation for combined theories.
\newblock In {\em CADE}, pages 183--198. Springer, 2009.

\bibitem{Griggio2012}
Alberto Griggio.
\newblock A practical approach to satisability modulo linear integer
  arithmetic.
\newblock {\em JSAT}, 8(1/2):1--27, 2012.

\bibitem{Griggio2011}
Alberto Griggio, Thi Thieu~Hoa Le, and Roberto Sebastiani.
\newblock Efficient interpolant generation in satisfiability modulo linear
  integer arithmetic.
\newblock In {\em TACAS}, pages 143--157. Springer, 2011.

\bibitem{HHP10}
Matthias Heizmann, Jochen Hoenicke, and Andreas Podelski.
\newblock Nested interpolants.
\newblock In {\em POPL}, pages 471--482. ACM, 2010.

\bibitem{henzinger04afp}
Thomas~A. Henzinger, Ranjit Jhala, Rupak Majumdar, and Kenneth~L. McMillan.
\newblock Abstractions from proofs.
\newblock In {\em POPL}, pages 232--244. ACM, 2004.

\bibitem{DBLP:journals/fmsd/JainCG09}
Himanshu Jain, Edmund~M. Clarke, and Orna Grumberg.
\newblock Efficient craig interpolation for linear diophantine (dis)equations
  and linear modular equations.
\newblock {\em Formal Methods in System Design}, 35(1):6--39, 2009.

\bibitem{DBLP:conf/lpar/KroeningLR10}
Daniel Kroening, J{\'e}r{\^o}me Leroux, and Philipp R{\"u}mmer.
\newblock Interpolating quantifier-free {Presburger} arithmetic.
\newblock In {\em LPAR}, pages 489--503. Springer, 2010.

\bibitem{DBLP:conf/atva/LynchT08}
Christopher Lynch and Yuefeng Tang.
\newblock Interpolants for linear arithmetic in {SMT}.
\newblock In {\em ATVA}, pages 156--170. Springer, 2008.

\bibitem{DBLP:conf/tacas/McMillan04}
Kenneth~L. McMillan.
\newblock An interpolating theorem prover.
\newblock In {\em TACAS}, pages 16--30. Springer, 2004.

\bibitem{mcmillan06lai}
Kenneth~L. McMillan.
\newblock Lazy abstraction with interpolants.
\newblock In {\em CAV}, pages 123--136. Springer, 2006.

\bibitem{McMillan12iZ3}
Kenneth~L. McMillan.
\newblock Interpolants from {Z3} proofs.
\newblock In {\em FMCAD}, pages 19--27. FMCAD Inc., 2011.

\bibitem{DBLP:journals/toplas/NelsonO79}
Greg Nelson and Derek~C. Oppen.
\newblock Simplification by cooperating decision procedures.
\newblock {\em ACM Trans. Program. Lang. Syst.}, 1(2):245--257, 1979.

\bibitem{dpllt}
Robert Nieuwenhuis, Albert Oliveras, and Cesare Tinelli.
\newblock Abstract {DPLL} and abstract {DPLL} modulo theories.
\newblock In {\em LPAR}, pages 36--50. Springer, 2004.

\bibitem{DBLP:journals/jsyml/Pudlak97}
Pavel Pudl{\'a}k.
\newblock Lower bounds for resolution and cutting plane proofs and monotone
  computations.
\newblock {\em J. Symb. Log.}, 62(3):981--998, 1997.

\bibitem{Yorsh2005}
Greta Yorsh and Madanlal Musuvathi.
\newblock A combination method for generating interpolants.
\newblock In {\em CADE}, pages 353--368. Springer, 2005.

\end{thebibliography}

\end{document}